\documentclass[11pt]{article}%
\usepackage{graphics,mathrsfs}
\usepackage{amsmath}
\usepackage{booktabs}
\usepackage{mathtools}
\usepackage{enumitem}
\usepackage{epsfig}
\usepackage{graphicx}
\usepackage{amssymb}
\usepackage{amsmath,bm}
\usepackage{caption}
\usepackage{subcaption}
\usepackage{subfloat}
\usepackage{bbold}
\usepackage{comment}

\newtheorem{theorem}{\bf Theorem}[section]
\newtheorem{proposition}{\bf Proposition}[section]

\newtheorem{corollary}{\bf Corollary}[section]
\newtheorem{remark}{\bf Remark}[section]
\newtheorem{definition}{\bf Definition}[section]

\newtheorem{example}{\bf Example }[section]

\usepackage{color}

\newenvironment{proof}{
\begin{trivlist}
\item[\hspace{\labelsep}{\bf\noindent Proof. }] }{\par\hfill $\Box$\end{trivlist}
\par}

\date{\plain}

\title{
\huge\bf Cumulative Information Generating Function and Generalized Gini Functions
\date{September 19, 2023}
}

\author{
\large   
Marco {\bf Capaldo}\thanks{ORCID: 0000-0002-0255-8935} \qquad
Antonio {\bf Di Crescenzo}\thanks{ORCID: 0000-0003-4751-7341} \qquad 
Alessandra {\bf Meoli} \thanks{ORCID: 0000-0002-3516-9530}\\
\\
\\
\normalsize Dipartimento di Matematica, Universit\`a degli Studi di Salerno\\ 
\normalsize Via Giovanni Paolo II, 132, I-84084 Fisciano (SA), Italy \\
\normalsize Email:  \{mcapaldo, adicrescenzo, ameoli\}@unisa.it  
}

\begin{document}
 
\maketitle

\abstract{
We introduce and study the cumulative information generating function, which 
provides a unifying mathematical tool suitable to deal with classical and fractional 
entropies based on the cumulative distribution function and on the survival function. 
Specifically, after establishing its main properties and some bounds, we show that 
it is a variability measure itself that extends the Gini mean semi-difference. 
We also provide (i) an extension of such a measure, based on distortion functions, 
and (ii) a weighted version based on a mixture distribution. 
Furthermore, we explore some connections with the reliability of $k$-out-of-$n$ 
systems and with stress-strength models for multi-component systems. 
Also, we address the problem of extending the cumulative information 
generating function to higher dimensions.

}


\section{Introduction and background}
In recent years, there is a deep interest in proposing new measures of uncertainty, in order to respond to the increasingly diversified needs of researchers in the fields of reliability and risk analysis. At the same time, getting lost in the vast sea of the new notions is relatively easy. 
To meet these needs, in this paper we aim to propose a new generating function which is able to recover the cumulative residual entropy and the cumulative entropy, as well as both their generalized and  fractional extensions. 
\par 

If $X$ is a nonnegative absolutely continuous random variable having support $(0,r)$, with $r\in(0,+\infty]$, 
and probability density function (PDF) $f$,  the \emph{differential entropy} of $X$ is defined as (see, for instance, Cover and Thomas \cite{CoverThomas})
\begin{equation}\label{differential entropy}
	\mathcal{H}(X)= \mathbb{E}\left[-\log  f(X)\right]= -\int_{0}^{r} f(x)\log f(x)\,\text{d}x.
\end{equation}
Such a measure can be obtained from the \emph{information generating function}, 
defined by Golomb \cite{Golomb:1966} as
\begin{equation}\label{eq:defIGX}
	\mathcal{IG}_X(\nu) = \mathbb{E}[ (f(X))^{\nu-1} ] = \int_0^r [f(x)]^{\nu} \,\text{d}x,
\end{equation}
where $\nu\in \mathbb{R}$ is such that the right-hand-side of (\ref{eq:defIGX}) is finite. Indeed,  Eqs.\ 
(\ref{differential entropy}) and (\ref{eq:defIGX}) give:
\begin{equation}
	\mathcal{H}(X)= - \frac{\text{d}}{\text{d}\nu}  \mathcal{IG}_X(\nu)\big\vert_{\nu=1}.
	\label{eq:relHX}
\end{equation}
Recent developments and examples of applications of information generating functions 
can be found in Clark \cite{Clark:2020}, Kharazmi et al.\ \cite{kharazmi:etal}, and Kharazmi and Balakrishnan \cite{khar:balak}. 
\par	
We remark that the differential entropy can take negative values, whereas the Shannon entropy of a discrete distribution 
is  nonnegative. To avoid this drawback, and for other reasons as mentioned in Rao et al.\ \cite{rao:etal}, 
various alternative measures have been proposed  recently. 
Table \ref{Entropies} shows some information measures 
for a random variable $X$ having support $(0,r)$, with $r\in(0,+\infty]$, 
having respectively cumulative distribution function  (CDF)  and survival function (SF) given by 
$$
F(x)=\mathbb{P}(X\le x), \qquad 
\overline{F}(x)= 1-F(x). 
$$
We  remark that the entropies presented in Table \ref{Entropies} can also be expressed in terms of the cumulative 
hazard rate and the cumulative reversed hazard rate of $X$, defined respectively as
\begin{equation}
	\Lambda(x)=-\log\overline{F}(x),  
	\qquad 
	T(x)=-\log F(x).
	\label{eq:defLaTx}
\end{equation}
These functions are involved in the  \emph{cumulative residual entropy} introduced in \cite{rao:etal}, 
and the \emph{cumulative entropy} (see Di Crescenzo and Longobardi \cite{dicre:longo2}, \cite{dicre:longo}), 
given respectively in cases (i) and (ii) of Table \ref{Entropies}. 
Such measures are obtained by replacing the PDF in (\ref{differential entropy}) with the SF and the CDF, respectively.
This preserves the fact that the logarithm of the probability of an event represents the information contained in the event, 
in accordance with the Shannon entropy in the discrete case. 
%
\begin{table}[t]
\caption{Information measures of interest, for a given random lifetime $X$ with support $(0,r)$ where $r\in(0,+\infty]$, with $n\in\mathbb{N}_0$ for case (iii), $n\in\mathbb{N}$ for case (iv), $\nu\ge0$ for case (v) and $\nu>0$ for case (vi).\label{Entropies}}
		\scalebox{0.85}{
			\begin{tabular}{l@{\hspace{0,5cm}}l}
				\toprule
				(i) cumulative residual entropy& (ii) cumulative entropy\vspace{0.3cm}\\
				$\mathcal{CRE}(X)=-\int_{0}^{r}\overline{F}(x)\log\overline{F}(x)\,\text{d}x$ 
				&$\mathcal{CE}(X)
				=-\int_{0}^{r}{F}(x)\log{F}(x) \,\text{d}x$
				\\
				\midrule
				(iii) generalized cumulative residual entropy&(iv) generalized cumulative entropy
				\vspace{0.3cm}\\
				$\mathcal{CRE}_n(X)=\frac{1}{n!}\int_{0}^{r}\overline{F}(x)\left[-\log\overline{F}(x)\right]^n \text{d}x$&$\mathcal{CE}_n(X)=\frac{1}{n!}\int_{0}^{r}F(x)\left[-\log F(x)\right]^n\, \text{d}x$\\
				\midrule
				(v) generalized fractional cumulative residual entropy&(vi) generalized fractional cumulative entropy\vspace{0.3cm}\\
				$\mathcal{CRE}_\nu(X)=\frac{1}{\Gamma(\nu+1)}\int_{0}^{r}\overline{F}(x)[-\log \overline{F}(x)]^\nu \,\text{d}x$&$	\mathcal{CE}_\nu(X)=\frac{1}{\Gamma(\nu+1)}\int_{0}^{r}F(x)[-\log F(x)]^\nu \,\text{d}x$\\
				\bottomrule
			\end{tabular}
	}
\end{table}
\par
Both the cumulative residual entropy and the cumulative entropy take nonnegative values, vanishing only in the case of degenerate random variables.
These measures are particularly suitable for describing information in problems related to reliability theory, where $X$ denotes the random lifetime of an item, and $x$ is the reference time. 
In particular, in Table \ref{Entropies}, the cumulative residual entropy (i) and the generalized versions 
(iii) and (v) deal with events for which the uncertainty is related to the future, 
while the cumulative entropies (ii), (iv) and (vi) are suitable to quantify the information 
when the uncertainty is related to the past. In addition, Asadi and Zohrevand \cite{Asadi:2007} 
showed that $\mathcal{CRE}(X)=\mathbb{E}[\text{mrl}(X)]$, where mrl$(X)$ is the mean residual life of $X$. Also, 
in \cite{dicre:longo} it is shown that $\mathcal{CE}(X)=\mathbb{E}[\tilde{\mu}(X)]$, where $\tilde{\mu}(X)$ is the mean inactivity time of $X$. 
Moreover, other applications of these information measures can be found in Risk Theory, since the risk is strictly related 
to the notion of uncertainty, see e.g.\ Dulac and Simon \cite{dulac}. 
\par
Recently, Psarrakos and Navarro \cite{Psarrakos:Navarro} introduced the \emph{generalized cumulative residual entropy of order $n$} of $X$, 
defined as in (iii) of Table \ref{Entropies}, in order to extend the cumulative residual entropy.
A dual information measure, known as \emph{generalized cumulative entropy of order $n$}, was proposed by Kayal \cite{Kayal:2016} 
(cf.\ case (iv) of Table \ref{Entropies}). Various results on these generalized measures have been studied by Toomaj and Di Crescenzo in \cite{dicre:toomaj}. 
In particular,  the measures given in  (iii) and (iv) of Table \ref{Entropies} play a role in the theory of point processes. 
Indeed, the generalized cumulative residual entropy of order $n$, say $\mathcal{CRE}_n(X)$, is equal to the mean of the 
$(n+1)$-th interepoch interval of a non-homogeneous Poisson process having cumulative intensity function given by the first of (\ref{eq:defLaTx}). 
Similarly, the generalized cumulative entropy of order $n$ can be viewed as an expected spacing in lower record values (see, 
for instance, Section 6 of \cite{dicre:toomaj}). 
They are also related to the upper and  lower record values densities (see, for instance, 
Kumar and Dangi \cite{KumarDangi:2022}). 
\par
Fractional versions of the above measures have been studied as well, with the aim of disposing with more advanced
mathematical tools  to handle complex systems and anomalous dynamics. Specifically, see Xiong et al.\ \cite{Xiong:2019} 
and Di Crescenzo et al.\ \cite{dicre:meoli}, for the \emph{fractional generalized cumulative residual entropy} 
and the \emph{fractional generalized cumulative entropy} of $X$, 
given respectively in cases (v) and  (vi) of Table \ref{Entropies}. 
Certain features of fractional calculus allow these measures to better capture 
long-range phenomena and nonlocal dependence in some random systems.
\par
The entropies considered in Table \ref{Entropies} deal with nonnegative 
random variables since they are often referred to random lifetimes of interest in reliability theory. 
However they can be straightforwardly extended to the case when $X$ 
has a general support contained in $(l,r)$, with $-\infty\leq l<r\leq +\infty$. 
\par
The aim of this paper is to propose and study a new information generating function that, in analogy with 
the functions in Eqs.\ 
(\ref{eq:defIGX}) and (\ref{eq:relHX}), is able to recover the information measures presented in Table \ref{Entropies}.  
It is defined as the integral of the product between suitable powers of the CDF  and the SF. 
Throughout the paper it emerges that some advantages related to the use of the new  generating function are: 
\\
-- \ the  convenience of gaining information from both the CDF  and the SF of the random variable under investigation,
\\
-- \ the existence of suitable applications to notions of interest in reliability theory such as proportional hazards, odds function, 
order statistics, $k$-out-of-$n$ systems, and stress-strength models for multi-component systems, 
\\
-- \ the possibility of using it as a measure of concentration, since it is an extension of the Gini mean semi-difference. 
\par
With reference to the latter statement, we will show also that the proposed generating function can be  extended 
(i) by replacing the powers of the CDF  and the SF with suitable distortion functions, and (ii) by defining a weighted 
version based on a mixture distribution. Moreover, it is worth mentioning that the proposed generating function 
and its generalized versions can be applied in risk analysis. Indeed, we show that they are proper variability measures.  
 
\subsection{Plan of the paper}
In Section \ref{section CIGF} we define the new generating function, named cumulative information generating function. 
We show that it is useful to recover the measures given in Table \ref{Entropies}. 
Moreover, we illustrate the effect of an affine transformation of the considered random variable, and provide 
some connections with the proportional  hazard model, the proportional reversed hazard model, and 
the odds function. 
\par
In Section \ref{sect:inequalities} we use various well-known inequalities in order to obtain some bounds for the 
cumulative information generating function. In addition, we show how the cumulative information generating function is related (i) to the Euler beta function, 
and (ii) to the Golomb's information generating function of the equilibrium random variable. 
\par
In Section \ref{section connections} we discuss the connections with some notions of systems reliability, 
as series and parallel systems, and $k$-out-of-$n$ systems, also with special attention to   
the reliability of the multi-component stress-strength system. 
\par
In Section \ref{sect:Generalizations}  we introduce the above mentioned generalized information measures, 
named `{\bf q}-distorted Gini function' and  `weighted {\bf q}-distorted Gini function', being related also to the Gini mean semi-difference. 
We also prove that they are suitable variability measures, since in particular the dispersive order between pairs of random variables 
implies the ordering between these functions. An application to   the reliability of  
multi-component stress-strength systems is  provided, too. 
\par
The Section \ref{section bidimensional CIGF} is concerning the extension of the 
cumulative information generating function to the case of a two-dimensional random vector, with special 
care to the case of independent components. 
\par
Some final remarks are then given in Section \ref{sect:final}. 
\par
Throughout the paper, 
the terms increasing and decreasing are used in non-strict sense, 
$\mathbb{N}$ denotes the set of positive integers, and $\mathbb{N}_0=\mathbb{N}\cup\{0\}$. 
Moreover, given a distribution function $F(x)$, we denote the right-continuous version of its inverse by 
$F^{-1}(u)=\sup\{x\colon F(x)\leq u\}$, $u\in [0,1]$, which is  also named  {\em quantile function} 
in statistical framework. 
%

\section{Cumulative information generating function}\label{section CIGF}
In the same spirit of Eq.\ (\ref{eq:defIGX}) we now introduce a new generating function 
which allows to measure the cumulative information coming both from the CDF and the SF.
\begin{definition}\label{def G_X}
	Let $X$ be a random variable with CDF $F(x)$ and SF $\overline{F}(x)$, $x\in\mathbb{R}$, and let  
	\begin{equation}\label{eq:deflr}
		l=\inf\{x\in\mathbb{R} : F(x)>0\}, \qquad r=\sup\{x\in\mathbb{R} : \overline{F}(x)>0\}
	\end{equation}
denote respectively the lower and upper limits of the support of $X$ (which may be finite or infinite). 
	The cumulative information generating function (CIGF) of $X$ is defined as
	\begin{equation}\label{G_X}
		\begin{aligned}
			G_X:&\:D_X\subseteq \mathbb{R}^{2} \longrightarrow [0,+\infty)\\
			&\;\;(\alpha,\beta)\quad\;\longmapsto \;G_X(\alpha,\beta)
			=\int_{l}^{r}\left[F(x)\right]^{\alpha}\left[\overline{F}(x)\right]^{\beta}\,{\rm{d}}x
		\end{aligned}
	\end{equation}
	where 
	$$
	D_X=\{(\alpha,\beta)\in\mathbb{R}^2 :   G_X(\alpha,\beta) <+\infty\}.
	$$
\end{definition}
Clearly, one has $\mathbb{P}(l\le X\le r)=1$. Moreover, if $X$ is degenerate then $G_X(\alpha,\beta)=0$, otherwise 
$G_X(\alpha,\beta)>0$. 
\begin{example}\label{example Erlang}
	Let $X\sim Erlang(2,\lambda)$, with $\lambda>0$, and  CDF  	
$F(x)=1-e^{-\lambda x}-\lambda xe^{-\lambda x}$, $x\ge0$. From 
	%
(\ref{G_X}), recalling that $\forall x : \vert x\vert<1$,
	\begin{equation}\label{useful for applications}
		(1+x)^{\alpha}=\sum_{n=0}^{+\infty}\binom{\alpha}{n}x^n, 
		\qquad \hbox{with }\; \binom{\alpha}{n}= \frac{\alpha(\alpha-1) \cdots (\alpha-n+1)}{n!},
	\end{equation}
	%
	and taking into account that 
	\begin{equation*}
		\int_{0}^{+\infty}e^{-(n+\beta)\lambda x}(1+\lambda x)^{n+\beta}\,\text{d}x
		=\frac{1}{\lambda}\,\Gamma(n+\beta+1,n+\beta)\,e^{(n+\beta)}(n+\beta)^{-(n+\beta+1)},
	\end{equation*} 
with $D_X=\{(\alpha,\beta)\in\mathbb{R}^2: \beta\in\mathbb{R}\setminus \mathbb{Z}_0^-\}$,
	we obtain the CIGF of $X$ in series form:
	\begin{equation}\label{CIGF of Erlang distribution}
		G_X(\alpha,\beta)=\frac{1}{\lambda}\sum_{n=0}^{+\infty}\binom{\alpha}{n}(-1)^n\,\Gamma(n+\beta+1,n+\beta)\,e^{n+\beta}\,(n+\beta)^{-(n+\beta+1)}.
	\end{equation}
\end{example}
\par 
We remark that if $X$ is a discrete random variable with finite support $\{x_{1}\leq x_{2}\leq \ldots \leq x_{n}\}$,
then due to (\ref{G_X}) the CIGF can be expressed as a sum, i.e. 
	$$
	G_X(\alpha,\beta)=\sum_{k=1}^{n-1} \left(P_k\right)^{\alpha}\left(1-P_k\right)^{\beta} (x_{k+1}-x_k), 
	\qquad \hbox{where \ } P_k=\sum_{i=1}^{k}\mathbb P(X=x_i). 
	$$
Other examples will be illustrated below. 
\par
In the next theorem we show the effect of an affine transformation. 
The result follows from Definition \ref{def G_X}, and recalling the relation between the CDFs of $X$ and $Y=\gamma X + \delta$. 
\begin{theorem}\label{theorematransformation}
	Let $X$ be a random variable with finite CIGF. 
	Consider the affine transformation $Y=\gamma X + \delta$, with $\gamma\in\mathbb{R}\setminus\{0\}$, $\delta\in\mathbb{R}$. Then
	\begin{equation}\label{Glineartrasformation}
		G_Y(\alpha,\beta)=
		\begin{cases}
			\gamma\,G_X(\alpha,\beta), & \text{$0<\gamma<\infty$} \\
			\vert\gamma\vert \, G_X(\beta,\alpha), & \text{$-\infty<\gamma<0$} 
		\end{cases}	
	\end{equation}
	for $(\alpha,\beta)\in D_X$ if $0<\gamma<\infty$, and 
	$(\beta,\alpha)\in D_X$ if $-\infty<\gamma<0$. 
\end{theorem}	
\begin{remark}
	If $X$ is absolutely continuous, with PDF $f(x)$, by setting $u=F(x)$ in the right-hand-side of 
	Eq.\ (\ref{G_X}), the CIGF of $X$ can be expressed as
	\begin{equation}\label{CIGF in Remark 1}
		G_X(\alpha,\beta)=\int_{0}^{1} u^{\alpha}(1-u)^{\beta}\frac{1}{f(F^{-1}(u))}\,{\rm d}u,
		\qquad (\alpha,\beta)\in D_X.
	\end{equation}
\end{remark}

\begin{remark}\label{rem:Gini}
		The CIGF of a nonnegative random variable $X$
		can be regarded as a measure of concentration. Indeed, if $ X' $ is an independent copy of $ X $, 
		we can deduce that (see, for instance,  the proof of Proposition 1 of Rao \cite{rao:05}) 
\begin{equation}
			\label{eq:GiniX}
			G_X(1,1)=\int_{l}^{r}F(x)\overline{F}(x)\,{\rm{d}}x= \frac12\, \mathbb{E}\left[\vert X-X'\vert\right]
			=Gini(X).
\end{equation}
		This quantity is also known as the Gini mean semi-difference,   which 
		represents an example of coherent measure of variability with comonotonic additivity 
		(see Section 2.2 of Hu and Chen \cite{Hu:Chen}).
\end{remark}

Similarly as the information generating function defined in (\ref{eq:defIGX}), 
the following generating measures can be introduced as marginal versions of the CIGF. 
\begin{definition}\label{remark H}
	Under the same assumptions of Definition \ref{def G_X}, 
	the cumulative information generating measure and 
	the cumulative residual information generating measure 
	are defined respectively by 
	\begin{equation}\label{H_X}
		H_X(\alpha) \equiv G_X(\alpha,0)=\int_{l}^{r}\left[F(x)\right]^\alpha\,{\rm{d}}x,
		\qquad \forall \; (\alpha,0)\in D_X
	\end{equation}
	and 
	\begin{equation}\label{CIG}
		K_X(\beta)\equiv G_X(0,\beta)=\int_{l}^{r}\left[\overline{F}(x)\right]^\beta\,{\rm{d}}x,
		\qquad \forall \; (0,\beta)\in D_X.
	\end{equation}
\end{definition}
We remark that when $X$ is absolutely continuous with support $(0,\infty)$, 
the measure (\ref{CIG}) has been introduced in Eq.\ (10) of  
Kharazmi and Balakrishnan \cite{kharazmi:balakrishnan}, denoted as $\mathcal{CIG}_\beta(\overline{F})$, 
for $\beta>0$. 
Under these assumptions, other properties and the non-parametric estimation of 
the function given in Eq.\ (\ref{CIG}) have been studied in Smitha et al.\ \cite{Smitha:etal}.

%

\begin{remark}\label{rem:symm}
	If $X$ is a random variable with finite CIGF and symmetric CDF, 
	in the sense that  for some $m\in\mathbb{R}$  one has
	$F(m+x)=\overline{F}(m-x)$ $ \forall \,x\in\mathbb{R}$, then
\\
(i) 	$G_X(\alpha,\beta)=G_X(\beta,\alpha)$  for all $(\alpha,\beta)\in D_X$; 
 \\
(ii) from Eqs.\ (\ref{H_X}) and (\ref{CIG})  we have 
	$H_X(\alpha)=K_X(\alpha)$  for all $(\alpha,0)\in D_X$; 
\\	
(iii) 
under the assumptions of Theorem \ref{theorematransformation}, 
Eq.\ (\ref{Glineartrasformation}) becomes 
	$ G_Y(\alpha,\beta)=\vert\gamma\vert\,G_X(\alpha,\beta)$, $  \gamma\in \mathbb R$.
\end{remark}
Recalling the measures (i) and (ii) of Table \ref{Entropies}, now 
we can show that the cumulative residual entropy and the cumulative entropy 
can be obtained from the CIGF.
\begin{proposition}\label{prop CE_X and CRE_X from G_X}
	Let $X$ be a random variable having finite CIGF $G_X(\alpha,\beta)$, for $(\alpha,\beta)\in D_X$. 
	If $(0,1)\in D_X$, then 
	\begin{equation*}
		\mathcal{CRE}(X)= - \frac{\partial}{\partial\beta}G_X(\alpha,\beta)\Big\vert_{\alpha=0,\beta=1}.
	\end{equation*}	
	If $(1,0)\in D_X$, then 
	\begin{equation*}
		\mathcal{CE}(X)= - \frac{\partial}{\partial\alpha}G_X(\alpha,\beta)\Big\vert_{\alpha=1,\beta=0}.
	\end{equation*}
\end{proposition} 
\begin{proof} 
	The stated results follow from Eq.\ (\ref{G_X}), by differentiation under the integral sign. 
\end{proof}	
Let us now obtain a similar relation for the 
generalized cumulative residual entropy and the generalized cumulative entropy 
(cf.\ cases (iii) and (iv) of Table \ref{Entropies}).
\begin{proposition}\label{prop CE_n and CRE_n from G_X}
	Let $X$ be a random variable having finite CIGF $G_X(\alpha,\beta)$, for $(\alpha,\beta)\in D_X$. 
	If $(0,1)\in D_X$, then 
	\begin{equation}\label{CRE_n from G_X}
		\mathcal{CRE}_n(X)
		= \frac{(-1)^n}{n!} \frac{\partial^n}{\partial\beta^n}G_X(\alpha,\beta)\Big\vert_{\alpha=0,\beta=1},
		\qquad n\in \mathbb{N}_0.
	\end{equation}	
	If $(1,0)\in D_X$, then 
	\begin{equation}\label{CE_n from G_X}
		\mathcal{CE}_n(X)
		= \frac{(-1)^n}{n!} \frac{\partial^n}{\partial\alpha^n}G_X(\alpha,\beta)\Big\vert_{\alpha=1,\beta=0},
		\qquad n\in \mathbb{N}.		
	\end{equation}
\end{proposition}
\begin{proof}
	The proof of (\ref{CE_n from G_X}) is analogous to Proposition \ref{prop CE_X and CRE_X from G_X}, by using this identity:
	\begin{equation}\label{useful for fractional case}
		\frac{\partial^n}{\partial \alpha^n}G_X(\alpha,\beta)
		=\int_{l}^{r}\left(\log F(x) \right)^n\left[F(x)\right]^\alpha\left[\overline{F}(x)\right]^\beta\,\text{d}x.
	\end{equation}
	In the same way we obtain Eq.\ (\ref{CRE_n from G_X}).
\end{proof}
In order to extend the above relations to the case of the generalized fractional cumulative residual entropy
and the generalized fractional cumulative entropy, given respectively  in cases (v) and (vi) of Table \ref{Entropies}, 
let us now recall briefly the expression of the Caputo fractional derivatives 
(see, for instance, Kilbas et al.\ \cite{North:Holland}).  
Specifically, given a function $y(x_1,x_2)$, we consider 
\begin{equation*}
	\left(\prescript{C}{}{D_{-,x_1}^\nu y}\right)(x_1,x_2)=\frac{(-1)^n}{\Gamma(n-\nu)}\int_{x_1}^{+\infty}\frac{y^{(n)}(t,x_2)}{(t-x_1)^{\nu+1-n}}\, \text{d}t, 
	\qquad x_1,x_2\in\mathbb{R},
\end{equation*}
that is the \emph{left-sided Caputo partial fractional derivative with respect to $x_1$ of order $\nu$} on the whole axis $\mathbb{R}$, where $\nu\in\mathbb{C}$ with $\text{Re}(\nu)>0$, $\nu\notin\mathbb{N}$ and $n=\lfloor \text{Re}(\nu)\rfloor + 1$.
\begin{proposition}\label{prop fractional CRE CE from CIGF }
	Let $X$ be a random variable having finite CIGF $G_X(\alpha,\beta)$, for $(\alpha,\beta)\in D_X$. 
	If $(0,1)\in D_X$, then 
	\begin{equation}\label{CRE_n fractional from G}
		\mathcal{CRE}_\nu(X)=\frac{1}{\Gamma(\nu+1)}\left(\prescript{C}{}{D_{-,\beta}^\nu G_X}\right)(\alpha,\beta)\Big\vert_{\alpha=0,\beta=1},
		\qquad \nu>0.
	\end{equation}
	If $(1,0)\in D_X$, then 	
	\begin{equation}\label{CE_n fractional from G}
		\mathcal{CE}_\nu(X)=\frac{1}{\Gamma(\nu+1)}\left(\prescript{C}{}{D_{-,\alpha}^\nu G_X}\right)(\alpha,\beta)\Big\vert_{\alpha=1,\beta=0},
		\qquad \nu>0.
	\end{equation}
\end{proposition}
\begin{proof} 
	We show only the proof of Eq.\ (\ref{CE_n fractional from G}) because Eq.\ (\ref{CRE_n fractional from G}) can be derived similarly. 
	From (\ref{useful for fractional case}) we obtain
	\begin{equation*}
		\begin{aligned}
			& \left( \prescript{C}{}{D_{-,\alpha}^\nu   G_X}\right)(\alpha,\beta)=
			\\
			&=\frac{(-1)^n}{\Gamma(n-\nu)}\int_{\alpha}^{+\infty}\frac{1}{(t-\alpha)^{\nu+1-n}}\frac{\partial^n}{\partial t^n}G_X(t,\beta)\,\text{d}t\\
			&=\frac{(-1)^n}{\Gamma(n-\nu)}\int_{\alpha}^{+\infty}\frac{1}{(t-\alpha)^{\nu+1-n}}
			\left(\int_{l}^{r}\left(\log F(x)\right)^n\left[F(x)\right]^t\left[\overline{F}(x)\right]^\beta \,\text{d}x\right)\,\text{d}t\\
			&=\frac{(-1)^n}{\Gamma(n-\nu)}\int_{l}^{r}
			\left(\log F(x)\right)^n\left[\overline{F}(x)\right]^\beta
			\left(\int_{\alpha}^{+\infty}\frac{\left[F(x)\right]^t}{(t-\alpha)^{\nu+1-n}}\,\text{d}t\right)\,\text{d}x,
		\end{aligned}	
	\end{equation*}
	where the last equality is obtained by use of Fubini's theorem.
	By placing $t-\alpha=z$ and $\gamma=-z\log  F(x)$, we have
	\begin{equation*}
		\begin{aligned}
			\int_{\alpha}^{+\infty}\frac{\left[F(x)\right]^t}{(t-\alpha)^{\nu+1-n}}\, \text{d}t
			&=\left[F(x)\right]^\alpha\int_{0}^{+\infty}\left[F(x)\right]^z z^{n-\nu-1} \, \text{d}z
			\\
			&=\left[F(x)\right]^\alpha\int_{0}^{+\infty}e^{-\gamma}\gamma^{n-\nu-1}
			\left(-\log  F(x) \right)^{-n+\nu}\, \text{d}\gamma
			\\
			&=\left[F(x)\right]^\alpha\Gamma(n-\nu)\left(-\log  F(x)\right)^{-n+\nu}.
		\end{aligned}
	\end{equation*}
	Finally we deduce
	\begin{equation*}
		\begin{aligned}
			\left(\prescript{C}{}{D_{-,\alpha}^\nu G_X}\right)(\alpha,\beta)  
			=\int_{l}^{r}\left(-\log F(x)\right)^{\nu}\left[\overline{F}(x)\right]^\beta\left[F(x)\right]^\alpha\,\text{d}x, 
		\end{aligned}
	\end{equation*}
	so that Eq.\ (\ref{CE_n fractional from G}) follows by taking $\alpha=1$ and $\beta=0$.
\end{proof}
\begin{remark}\label{remark marginal}
	Recently, Kharazmi and Balakrishnan \cite{kharazmi:balakrishnan} noted that 
$\mathcal{CRE}(X)=-\frac{\text{d}}{\text{d}\beta}K_X(\beta)\big\vert_{\beta=1}$. 
	Similarly,  we obtain  
$\mathcal{CE}(X)=-\frac{{\rm d}}{{\rm d}\alpha}H_X(\alpha)\big\vert_{\alpha=1}$.
	Moreover, for  the generalized versions and for the fractional versions we have respectively
	\begin{equation*}
		\mathcal{CRE}_n(X)=\frac{(-1)^n}{n!}\frac{{\rm d}^n}{{\rm d}\beta^n}K_X(\beta)\Big\vert_{\beta=1},
		\qquad
		\mathcal{CE}_n(X)=\frac{(-1)^n}{n!}\frac{{\rm d}^n}{{\rm d}\alpha^n}H_X(\alpha)\Big\vert_{\alpha=1},
	\end{equation*}
	\begin{equation*}
		\mathcal{CRE}_\nu(X)=\frac{1}{\Gamma(\nu+1)}\left(\prescript{C}{}{D_{-}^\nu K_X}\right)(\beta)\Big\vert_{\beta=1},
		\qquad		
		\mathcal{CE}_\nu(X)=\frac{1}{\Gamma(\nu+1)}\left(\prescript{C}{}{D_{-}^\nu H_X}\right)(\alpha)\Big\vert_{\alpha=1}.
	\end{equation*}
\end{remark}
Now we recall two important models that are largely adopted in survival analysis and reliability theory.
Let $X$ be a random lifetime with CDF $F(x)$ and SF $\overline{F}(x)$. The proportional hazard model 
(see for instance Cox \cite{Cox:72}, Kumar and Klefsj\"o \cite{Kumar:1994}) is expressed by a random lifetime 
$X^*_\gamma$ with SF
\begin{equation}\label{propotional hazard model}
	\overline{F}_{X^*_\gamma}(x)=\left[\overline{F}(x)\right]^\gamma,
	\qquad
	\gamma\in\mathbb{R}^+.
\end{equation}
Similarly, the proportional reversed hazard model (see for instance Di Crescenzo \cite{dicre:2000}, 
Gupta and Gupta \cite{Gupta:2007}, Gupta et al.\ \cite{Gupta:1998}) is expressed by a random lifetime 
$\hat{X}_\theta$ with CDF
\begin{equation}\label{proportional reversed hazard model}
	F_{\hat{X}_\theta}(x)=\left[F(x)\right]^\theta, \qquad\theta\in\mathbb{R}^+.
\end{equation}
Recently, modified versions of these models have been studied by Das and Kayal \cite{Das:Kayal}.
\begin{remark}
	The measures given in Definition \ref{remark H} satisfy the following relations. 
	\begin{description}
		\item{(i)}
		Under the proportional hazard model (\ref{propotional hazard model}) we have 
		$$
		K_{X^*_\gamma}(\beta)=K_X(\gamma\beta)  
		\qquad  \forall \,\gamma\in\mathbb{R}^+ \;\; \hbox{s.t.}\;\; (0,\gamma\beta)\in D_X.
		$$
		\item{(ii)}
		Under the proportional reversed hazard model (\ref{proportional reversed hazard model}) we have 
		$$
		H_{\hat{X}_\theta}(\alpha)=H_X(\theta\alpha) 
		\qquad  \forall \,\theta\in\mathbb{R}^+ \;\; \hbox{s.t.}\;\; (\theta\alpha,0)\in D_X.
		$$   
	\end{description}
\end{remark}
\begin{remark}
	If $X$  is a random lifetime   such that $D_X\subseteq(\mathbb{R}^+)^2$, then recalling the Definition \ref{def G_X}  
	and Eqs.\ (\ref{propotional hazard model}) and   (\ref{proportional reversed hazard model}),   the CIGF of $X$ 
	can be expressed as  
	$$
	G_X(\alpha,\beta)=\int_{l}^{r}F_{\hat{X}_\alpha}(x)\,\overline{F}_{X^*_\beta}(x)\,\text{d}x. 
	$$ 
\end{remark}
\par
We now recall another  useful concept. 
Let $X$ be a random variable with CDF and SF denoted by $F(x)$ and  $\overline{F}(x)$, respectively. 
For all $x\in(l,r)$ 
the \emph{odds function} of $X$ is (cf.\ Kirmani and Gupta \cite{odds:model})
\begin{equation}\label{Odds Function}
	\theta(x)=\frac{\overline{F}(x)}{F(x)}.
\end{equation}
This function represents  the ratio of the probability of an event occurring to the probability of its not occurring, 
and always assumes nonnegative finite values. 
It is used in reliability theory, because it quantifies the strength of the association 
between the failure of a system after time $ x $ and before time $ x $. 
Due to Eq.\ (\ref{Odds Function}), we can express the CIGF of $X$ in terms of the odds function in 
two equivalent useful ways: 
\begin{equation*}
	G_X(\alpha,\beta)=\int_{l}^{r}[F(x)]^{\alpha+\beta}[\theta(x)]^{\beta}\,\text{d}x
	=\int_{l}^{r}\left[\theta(x)\right]^{-\alpha}\left[\overline{F}(x)\right]^{\alpha+\beta}\,\text{d}x.
\end{equation*}
Hence, when the parameters $\alpha$ and $\beta$ may take negative values such that $\alpha+\beta=0$, then 
the CIGF of $X$ can be expressed   in terms of the odds function as
\begin{equation}\label{G_X in terms of only Odds Function with beta}
	G_X(-\beta,\beta)=\int_{l}^{r}[{\theta(x)}]^{\beta}\,\text{d}x
	\qquad \forall \,\beta\in D_{X,\theta}
\end{equation}
where
$$
D_{X,\theta}=\{ \beta \in \mathbb{R} \colon  G_{X}(-\beta,\beta)   <+\infty\}.
$$
Table \ref{tableCIGF} shows various examples of the CIGF expressed in terms of the Euler Beta function 
$B(x,y)=\int_{0}^{1}t^{x-1}(1-t)^{y-1}\,\text{d}t$ or the incomplete Beta function 
$B\left(p;x,y\right)=\int_{0}^{p}t^{x-1}(1-t)^{y-1}\,\text{d}t$, $p\in [0,1]$. 
\par 
Finally, by recalling Eq.\ (\ref{eq:GiniX}), we remark that for $\alpha=\beta=1$ Eq.\ (\ref{CIGF of Erlang distribution}) 
and the examples in Table \ref{tableCIGF} are in agreement with Giorgi and Nadarajah \cite{Giorgi:Nadarajah}.
\begin{table}[!t]
		\caption{The CIGF for some notable distributions.\label{tableCIGF}}
		\scalebox{0.85}{
		{\small
			\begin{tabular}{lccc}
				\toprule
				Distribution &$D_X$	&$D_{X,\theta}$& $G_X(\alpha,\beta)$\\
				\midrule
				$X \sim \text{Bernoulli}(p)$&$\mathbb{R}^2$&$\mathbb{R}$ 
				&$(1-p)^\alpha\, p^\beta=\frac{{\rm d}}{{\rm d}p}B\left(p; \beta+1,\alpha+1\right)$\\
				\midrule
				$X\sim\text{Unif}(l,r)$&$\{(\alpha,\beta)\in \mathbb{R}^2:-1<\alpha<\infty, -1<\beta<\infty\}$&$(-1,1)$ &$(r-l)B(\alpha+1,\beta+1)$\\
				\midrule
				$X \sim \text{Power}(\theta)$&$\{(\alpha,\beta)\in \mathbb{R}^2:-\frac{1}{\theta}<\alpha<\infty, -1<\beta<\infty\}$&
				$\left(-\frac{1}{\theta},\frac{1}{\theta}\right)$ &$\frac{1}{\theta}\;B\left( \alpha+\frac{1}{\theta},\beta+1\right)$\\
				\midrule
				$X\sim \text{Exp}(\lambda)$&$\{(\alpha,\beta)\in \mathbb{R}^2:-1<\alpha<\infty, 0<\beta<\infty\}$&$(0,1)$ & 	$\frac{1}{\lambda}B(\alpha+1,\beta)$\\
				\midrule
				$X \sim \text{Laplace}(0,\lambda)$&$\{(\alpha,\beta)\in \mathbb{R}^2:0<\alpha<\infty,\;0<\beta<\infty\}$&$\emptyset$ &$\frac{\lambda}{2}\left[B\left(\frac12;\alpha,\beta+1\right)+B\left(\frac12;\alpha+1,\beta\right)\right]$\\
				\bottomrule	
		\end{tabular}
		}
		}
\end{table}


\section{Inequalities and further results}\label{sect:inequalities}
In this section, we obtain some bounds and further results regarding the CIGF. 
Specifically, we first refer to well-known inequalities named after 
Chernoff, Bernoulli,  Minkowski and H\"older's (see, for istance Schilling \cite{Schilling:05}). 
\par
Hereafter, thanks to the Chernoff's inequalities we recover some bounds for the CIGF in terms of   
the moment generating function (MGF) of $X$, denoted by $M_X(s)=\mathbb{E}(e^{sX})$, $s\in \mathbb R$. 
\begin{proposition}\label{prop Chernoff's inequalities}
	Let $X$ be a nonnegative random variable with support $(0,r)$, where $r\in(0,+\infty)$, and having finite CIGF $G_X(\alpha,\beta)$, 
	for $(\alpha,\beta)\in D_X$. Assume that the MGF of $X$ satisfies $M_X(s)<+\infty$ for all $s\in (-s_0, s_0)$, with $s_0>0$. 
	Then, 
	\\
	(i) for all  $s_1\in(-s_0,0)$, $s_2\in(0,s_0)$ and $(\alpha,\beta)\in D_X\cap(\mathbb{R}^+)^2$, one has 
	\begin{equation}\label{Chernoff's inequality for CIGF +}
		G_X(\alpha,\beta)
		\le g(r; \alpha,\beta, {\bf s}) \left[M_X(s_1)\right]^\alpha\left[M_X(s_2)\right]^\beta,
	\end{equation}
	(ii) for all  $s_1\in(-s_0,0)$, $s_2\in(0,s_0)$ and $(\alpha,\beta)\in D_X\cap(\mathbb{R}^-)^2$, one has 
	\begin{equation}\label{Chernoff's inequality for CIGF -}
		G_X(\alpha,\beta)\ge g(r; \alpha,\beta, {\bf s})
		\left[M_X(s_1)\right]^\alpha\left[M_X(s_2)\right]^\beta,
	\end{equation}
	where 
	$$
	g(r; \alpha,\beta, {\bf s})=\left\{
	\begin{array} {ll}
		\displaystyle 
		\frac{1}{\alpha s_1+\beta s_2}\left[1-e^{-(\alpha s_1+\beta s_2)r}\right], & \hbox{ if }\alpha s_1+\beta s_2\neq 0, \\[2mm]
		r, &  \hbox{ if }\alpha s_1+\beta s_2= 0.
	\end{array} 
	\right.
	$$
\end{proposition}
\begin{proof}
	Applying Chernoff's inequalities $\mathbb{P}(X\le x)\le e^{-s_1x}M_X(s_1)$ 
	and $ \mathbb{P}(X\ge x)\le e^{-s_2x}M_X(s_2)$, for $x>0$, 
	to Eq.\ (\ref{G_X}), by the effect of integration, for all  $s_1\in(-s_0,0)$, $s_2\in(0,s_0)$ and $(\alpha,\beta)\in D_X\cap(\mathbb{R}^+)^2$, it follows that
	$$
	G_X(\alpha,\beta)\le\int_{0}^{r}\left[e^{-s_1x}M_X(s_1)\right]^\alpha\left[e^{-s_2x}M_X(s_2)\right]^\beta\,\text{d}x.
	$$
	Few calculations give Eq.\ (\ref{Chernoff's inequality for CIGF +}). For $(\alpha,\beta)\in D_X\cap(\mathbb{R}^-)^2$, 
	Eq.\ (\ref{Chernoff's inequality for CIGF -}) can be obtained similarly.
\end{proof}
The case when $X$ has support $(0,+\infty)$ can be easily derived as follows. 
\begin{corollary}
	Let $X$ be a nonnegative random variable with support $(0,+\infty)$, and having finite CIGF $G_X(\alpha,\beta)$, 
	for $(\alpha,\beta)\in D_X$. Assume that the MGF of $X$ satisfies $M_X(s)<+\infty$ for all $s\in (-s_0, s_0)$, with $s_0>0$. 
	Then, 
	\\
	(i) for all  $s_1\in(-s_0,0)$, $s_2\in(0,s_0)$ and $(\alpha,\beta)\in D_X\cap(\mathbb{R}^+)^2$ such that $\alpha s_1+\beta s_2>0$, one has 
	\begin{equation}\label{Chernoff's inequality for CIGF + with r=inf}
		G_X(\alpha,\beta)\le\frac{1}{\alpha s_1+\beta s_2}\left[M_X(s_1)\right]^\alpha\left[M_X(s_2)\right]^\beta,
	\end{equation}
	(ii) for all  $s_1\in(-s_0,0)$, $s_2\in(0,s_0)$ and $(\alpha,\beta)\in D_X\cap(\mathbb{R}^-)^2$ such that $\alpha s_1+\beta s_2>0$, one has 
	\begin{equation*}
		G_X(\alpha,\beta)\ge\frac{1}{\alpha s_1+\beta s_2}\left[M_X(s_1)\right]^\alpha\left[M_X(s_2)\right]^\beta.
	\end{equation*}
\end{corollary}
\par
The following example provides an application of the previous results. 
\begin{example}
	Let $X\sim Erlang(2,\lambda)$,  $\lambda>0$, as in Example \ref{example Erlang}, with MGF 
	$M_X(s)=\lambda^2 (\lambda-s)^{-2}$, for $s<\lambda$. Thanks to Eqs.\ (\ref{CIGF of Erlang distribution}) 
	and  (\ref{Chernoff's inequality for CIGF + with r=inf}), 
	some calculations give, for $(\alpha,\beta)\in  (\mathbb{R}^+)^2$, 
$$
		G_X(\alpha,\beta) \le \inf_{\substack{(s_1,s_2)\in(-\lambda,0)\times(0,\lambda)\\ \alpha s_1+\beta s_2>0}}\,
		\frac{1}{\alpha s_1+\beta s_2}\cdot\frac{\lambda^{2(\alpha+\beta)}}{(\lambda-s_1)^{2\alpha}(\lambda-s_2)^{2\beta}}
		= \frac{1}{\lambda}\frac{1}{2^{2\beta}}\left(\frac{1+2\beta}{\beta}\right)^{1+2\beta}.
$$
\end{example}
\par
Let us now express some upper bounds for the CIGF in terms of the  measures introduced in 
Definition \ref{remark H}. 
\begin{proposition}
	Under the assumptions specified in Definition \ref{def G_X}, the CIGF of a random variable $X$ satisfies 
	the following inequalities:
	\begin{equation*}
		G_X(\alpha,\beta)\le K_X(\beta)-\alpha K_X(\beta+1)\qquad
		\forall\, (\alpha,\beta)\in D_X\cap[0,1]\times\mathbb{R},
	\end{equation*}
	\begin{equation*}
		G_X(\alpha,\beta)\le H_X(\alpha)-\beta H_X(\alpha+1)\qquad
		\forall\,(\alpha,\beta)\in D_X\cap\mathbb{R}\times[0,1].
	\end{equation*}
\end{proposition}
\begin{proof}
	Due to Bernoulli's inequality with real exponents, for all $x\in\mathbb{R}$ it follows that
	$$
	\left[1-\overline{F}(x)\right]^\alpha\le1 -\alpha\overline{F}(x)
	\qquad\forall\, \alpha\in[0,1]
	$$
	and 
	$$
	\left[1-F(x)\right]^\beta\le1-\beta F(x) 
	\qquad\forall\,\beta\in[0,1].
	$$
	Hence, the thesis immediately follows from Definitions \ref{def G_X} and \ref{remark H}.
\end{proof}
\par
Hereafter we use the Minkowski's inequality to obtain suitable bounds for the measures introduced in Definition \ref{remark H} 
and for $G_X(\gamma,\gamma)$. 
\begin{proposition}
	Under the assumptions specified in Definition \ref{def G_X} and Definition \ref{remark H}, if $X$ has finite support in $(l,r)$, 
	\\
	(i) for all $\gamma\ge1$ such that $(\gamma,0)\in D_X$ and $(0,\gamma)\in D_X$ we have 
	\begin{equation*}
		\left[\left(r-l\right)^\frac{1}{\gamma}-\left[H_X(\gamma)\right]^\frac{1}{\gamma}\right]^\gamma\le K_X(\gamma)
		\le\left[\left(r-l\right)^\frac{1}{\gamma}+\left[H_X(\gamma)\right]^\frac{1}{\gamma}\right]^\gamma,
	\end{equation*}
	\begin{equation*}
		\left[\left(r-l\right)^\frac{1}{\gamma}-\left[K_X(\gamma)\right]^\frac{1}{\gamma}\right]^\gamma\le H_X(\gamma)\le\left[\left(r-l\right)^\frac{1}{\gamma}+\left[K_X(\gamma)\right]^\frac{1}{\gamma}\right]^\gamma;
	\end{equation*}
	(ii)  for all $\gamma\ge1$ such that $(0,\gamma)\in D_X$  we have 
	\begin{equation*}
		G_X(\gamma,\gamma)\le\left[\left[K_X(\gamma)\right]^\frac{1}{\gamma}+\left[K_X(2\gamma)\right]^\frac{1}{\gamma}\right]^\gamma;
	\end{equation*}
	(iii)  for all $\gamma\ge1$ such that $(\gamma,0)\in D_X$  we have 
	\begin{equation*}
		G_X(\gamma,\gamma)\le\left[\left[H_X(\gamma)\right]^\frac{1}{\gamma}+\left[H_X(2\gamma)\right]^\frac{1}{\gamma}\right]^\gamma.
	\end{equation*}
\end{proposition}
\begin{proof}
	By applying Minkowski's inequality, for $\gamma\ge1$ we have
	\begin{equation*}
		\begin{aligned}
			\left[K_X(\gamma)\right]^\frac{1}{\gamma}
			=\left[\int_{l}^{r}\left[1-F(x)\right]^\gamma\, \text{d}x\right]^\frac{1}{\gamma}
			&\le\left(r-l\right)^\frac{1}{\gamma}+\left[\int_{l}^{r}\left[F(x)\right]^\gamma\, \text{d}x\right]^\frac{1}{\gamma}\\
			&=\left(r-l\right)^\frac{1}{\gamma}+\left[H_X(\gamma)\right]^\frac{1}{\gamma},
		\end{aligned}		
	\end{equation*}
	and also 
	\begin{equation*}
		\begin{aligned}
			\left(r-l\right)^\frac{1}{\gamma}
			=\left[\int_{l}^{r}\left[F(x)+\overline{F}(x)\right]^\gamma\, \text{d}x\right]^\frac{1}{\gamma}
			&\le\left[\int_{l}^{r}\left[F(x)\right]^\gamma\,\text{d}x\right]^\frac{1}{\gamma}
			+\left[\int_{l}^{r}\left[\overline{F}(x)\right]^\gamma\, \text{d}x\right]^\frac{1}{\gamma}\\
			&=\left[H_X(\gamma)\right]^\frac{1}{\gamma}+\left[K_X(\gamma)\right]^\frac{1}{\gamma}.
		\end{aligned}	
	\end{equation*}
	Combining the two latter inequalities we obtain the bounds for $K_X(\gamma)$. 
	The other relations can be obtained in the same way, by taking into account that, for $\gamma \geq 1$, if 
	$(0,\gamma)\in D_X$ then $(0,2\gamma)\in D_X$, and if $(\gamma,0)\in D_X$ then $(2\gamma,0)\in D_X$. 
\end{proof}
\par
We now prove an upper bound for $G_X(\alpha,\beta)$, for $\alpha+\beta=1$, making use of the H\"older's inequality. 
\begin{proposition}\label{prop Holder's inequality}
	Under the assumptions specified in Definition \ref{def G_X}, 
	let $X$ have finite support in $(l,r)$, with $(\theta,1-\theta)\in D_X$ for all $\theta\in(0,1)$.
	Then, 
	%
	%
	%
	\begin{equation}\label{Holder's inequality}
		G_X(\theta,1-\theta)\le\left[r-\mathbb{E}(X)\right]^\theta\left[\mathbb{E}(X)-l\right]^{1-\theta}
		\qquad\forall \; \theta\in(0,1).
	\end{equation}
\end{proposition}
\begin{proof}
	Due to the H\"older's inequality with conjugate exponents $\frac{1}{\theta},\frac{1}{1-\theta}$, for all $\theta\in(0,1)$ we have
	\begin{equation*}
		G_X(\theta,1-\theta)\le\left(\int_{l}^{r}F(x)\,\text{d}x\right)^\theta\left(\int_{l}^{r}\overline{F}(x)\,\text{d}x\right)^{1-\theta},
	\end{equation*}
	this yielding Eq.\ (\ref{Holder's inequality}). 
\end{proof} 
\par
We remark that Eq.\ (\ref{Holder's inequality}) is satisfied as equality when $F(x)=\overline{F}(x)$ $\forall \, x\in(l,r)$, i.e.\ when $\mathbb P(X=l)=\mathbb P(X=r)=1/2$. Moreover, we note that the right-hand-side of Eq.\ (\ref{Holder's inequality}) can be rewritten by taking into account that, under the given assumptions, 
$$
r-\mathbb{E}(X)=H_X(1), \qquad \mathbb{E}(X)-l=K_X(1).
$$ 
\par 
Hereafter we show that the CIGF of an absolutely continuous random variable  can be expressed as the product of the Euler beta function and the expected value of a suitably transformed beta-distributed random variable.
\begin{proposition}
	If $X$ is an absolutely continuous random variable with PDF $f$, CDF $F$ and finite CIGF, 
	then,
	$$
	G_X(\alpha,\beta)=B(\alpha+1,\beta+1)\,\mathbb{E}\left[r(Y)\right],
	$$
	where $Y\sim {\rm Beta}(\alpha+1,\beta+1)$ is independent from $X$, with 
	$r(Y)=[f(F^{-1}(Y))]^{-1}$. 
\end{proposition}
\begin{proof}
	The thesis is obtained making use of Eq.\ (\ref{CIGF in Remark 1}) and recalling the PDF of $Y\sim {\rm Beta}(\alpha+1,\beta+1)$.
\end{proof}
\par
The above result collects various features of the CIGF, i.e.\ the relations (i) to the Beta distribution,  which reflects the form of the right-hand-side 
of (\ref{G_X}), and (ii) to the transformation $f(F^{-1}(\cdot))$, which plays a relevant role in the context of variability measures   as 
developed in Section \ref{sect:Generalizations}. 
\par
We conclude this section by relating the CIGF to a series involving the Golomb's information generating function of the equilibrium random variable. 
To this aim, we recall that for a nonnegative random variable $X$, with SF $\overline{F}(x)$ and 
expected value $\mathbb{E}[X]\in (0, +\infty)$, the equilibrium random variable of $X$ is a nonnegative 
absolutely continuous random variable, denoted as $X_e$, whose PDF is given by 
\begin{equation}\label{f_e(x) of X_e}
	f_e(x)=\frac{\overline{F}(x)}{\mathbb{E}[X]}, \qquad x>0.
\end{equation}
Recalling Eq.\ (\ref{eq:defIGX}), hereafter we denote by $\mathcal{IG}_{X_e}$ the 
Golomb's information generating function  of the equilibrium random variable $X_e$.
\begin{proposition}\label{prop CIGF in terms of equilibrium PDF}
	If $X$ is a nonnegative random variable having  expected value $\mathbb{E}[X]\in (0, +\infty)$ 
	and with finite CIGF, then 
	\begin{equation*}
		G_X(\alpha,\beta)=\sum_{n=0}^{\infty}\binom{\alpha}{n}(-1)^n(\mathbb{E}[X])^{n+\beta}\,\mathcal{IG}_{X_e}(n+\beta)
		\qquad \forall \; (\alpha,\beta)\in D_X,
	\end{equation*}
	where $X_e$ is the equilibrium random variable of $X$.
\end{proposition}
\begin{proof}
	Denoting by $(0,r)$ the support of $X$, with $r\in(0,+\infty]$, 
	from Eqs.\ (\ref{G_X}) and (\ref{f_e(x) of X_e})  it follows that, for all $(\alpha,\beta)\in D_X$
	\begin{equation}\label{useful for CIGF in terms of equilibrium PDF}
		G_X(\alpha, \beta)
		=\int_{0}^{r}\left(1-f_e(x)\,\mathbb{E}[X]\right)^{\alpha}(f_e(x)\,\mathbb{E}[X])^{\beta}\;\text{d}x.
	\end{equation}
	Since $\left\vert f_e(x)\,\mathbb{E}[X]\right\vert <1$, due to Eq.\  (\ref{useful for applications}) we have 
	$$
	(1-f_e(x)\,\mathbb{E}[X])^{\alpha}=\sum_{n=0}^{\infty}\binom{\alpha}{n}(-1)^n\left(f_e(x)\,\mathbb{E}[X]\right)^n.
	$$	
	By replacing the latter equation in   (\ref{useful for CIGF in terms of equilibrium PDF}), 
	the thesis immediately follows from Eq.\ (\ref{eq:defIGX}).
\end{proof}
%
\section{Connections with systems reliability}\label{section connections}
In this section we relate some results  exploited above to notions of interest in reliability theory. 
\par
Several applied  problems  involve complex systems consisting of many components. 
Here we focus on systems formed by $n$ components, where $X_1,X_2,\dots,X_n$ describe the random lifetimes 
of each component. We assume that they are independent and identically distributed (i.i.d.), with common CDF $F(x)$ and SF $\overline{F}(x)$. 
As well known,  a parallel system continues to work until the last component fails, and thus 
its lifetime is described by the sample maximum  
$$
X_{(n:n)}=\max\{X_1,X_2,\dots,X_n\},
$$
which has CDF
\begin{equation}\label{CDF of maximum order statistic}
	F_{(n:n)}(x)=\mathbb{P}(X_{(n:n)}\le x)=\left(F(x)\right)^n, 
	\qquad x\in \mathbb R. 
\end{equation}
Similarly, a series system  fails as soon as the first component stops working, and thus 
its lifetime is described by the sample minimum  
$$
X_{(1:n)}=\min\{X_1,X_2,\dots,X_n\},
$$
that possesses SF  
\begin{equation}\label{SF of minimum order statistic}
	\overline{F}_{(1:n)}(x)=\mathbb{P}(X_{(1:n)}>x)=\left(\overline{F}(x)\right)^n, 
	\qquad x\in \mathbb R.
\end{equation}
\begin{remark}
	Let $n\in \mathbb N$. 
	Recalling Definition \ref{remark H}, from Eqs.\  (\ref{CDF of maximum order statistic}) and
	(\ref{SF of minimum order statistic}) it immediately follows that
	$$
	H_{X_{(n:n)}}(\alpha)=H_X(n\alpha), \qquad \forall \; (n\alpha,0)\in D_X,
	$$
	$$
	K_{X_{(1:n)}}(\beta)=K_X(n\beta), \qquad \forall \; (0,n\beta)\in D_X,
	$$
	where $H_X$  and $K_X$ denote respectively the cumulative  information generating measure 
	and the cumulative residual information generating measure of $X_i$.
\end{remark}
We now focus on the expression of the CIGF for order statistics $X_{(n:n)}$ and $X_{(1:n)}$.
\begin{proposition}\label{prop G of minimum and maximum order statistcs}
	For $n\in \mathbb N$, let  $X_1,X_2,\dots,X_n$  be a random sample formed by  
	i.i.d.\ random lifetimes having finite cumulative  information generating measure  $H_X$  
	and cumulative residual information generating measure $K_X$. Then,   
	the CIGF of the order statistics $X_{(n:n)}$ and $X_{(1:n)}$ can be expressed respectively as 
	\begin{equation}\label{G for maximum order statistic}
		G_{X_{(n:n)}}(\alpha,\beta)
		=\sum_{i=0}^{\infty}(-1)^i\binom{\beta}{i}H_{X}\left( n(i+\alpha)\right), 
		\qquad \forall \; (\alpha, \beta)\in D_{X_{(n:n)}}
	\end{equation}
	and
	\begin{equation}\label{G for minimum order statistic}
		G_{X_{(1:n)}}(\alpha,\beta)
		=\sum_{j=0}^{\infty}(-1)^j\binom{\alpha}{j}K_{X}\left( n(j+\beta)\right),
		\qquad \forall \; (\alpha, \beta)\in D_{X_{(1:n)}}.
	\end{equation}
\end{proposition}
\begin{proof} 
	For simplicity, assume that the support of $X$ is $(0, r)$. 
	Recalling Eq.\ (\ref{G_X}), from Eqs.\  (\ref{CDF of maximum order statistic}) and (\ref{useful for applications}) we have 
$$
 G_{X_{(n:n)}}(\alpha,\beta)=\int_{0}^{r}\left[F(x)\right]^{n\alpha}\left[1-[ {F}(x)]^n\right]^{\beta}\,\text{d}x 
	=\sum_{i=0}^{\infty}(-1)^i\binom{\beta}{i}\int_{0}^{r}\left[ {F}(x)\right]^{n(i+\alpha)}\,\text{d}x.
$$
	The right-hand-side of (\ref{G for maximum order statistic}) then follows making use of (\ref{H_X}). 
	Eq.\ (\ref{G for minimum order statistic}) can be obtained similarly. 
\end{proof}
\par
It is well known that a system with $n$ independent components is said to be a \emph{$k$-out-of-$n$ system} when 
it works if and only if at least $k$ components work (see, for instance, Boland and Proschan \cite{Boland:Proschan}).
Clearly, if $k=1$ we have a parallel system, while for $k=n$ we have a series system.  
For any $k$, the lifetime of the $k$-out-of-$n$ system formed by components with  i.i.d.\ lifetimes 
is expressed as the corresponding $k$-th order statistic. This allows to express the reliability and the information 
content of this kind of systems in a tractable way. 
\begin{remark}
	Consider a $k$-out-of-$n$ system formed by $n$ components with  i.i.d.\ random lifetimes, for $n\in\mathbb{N}$. 
	Denoting by $X_{(k:n)}$   the corresponding $k$-th order statistic, which in turn gives the system lifetime, the 
	CIGF of $X_{(k:n)}$ can be expressed in terms of the cumulative  information generating measures. Indeed, similarly as 
	Proposition \ref{prop G of minimum and maximum order statistcs}, recalling Eqs.\ (\ref{H_X}) and  (\ref{CIG}) 
	one has the following two equivalent expressions
	$$
	G_{X_{(k:n)}}(\alpha,\beta)=\sum_{i=0}^{\infty}(-1)^i\binom{\beta}{i}H_{X_{(k:n)}}(i+\alpha),
	$$	
	$$
	G_{X_{(k:n)}}(\alpha,\beta)=\sum_{j=0}^{\infty}(-1)^j\binom{\alpha}{j}K_{X_{(k:n)}}(j+\beta),
	$$	
	for all $(\alpha,\beta)\in D_{X_{(k:n)}}$. Moreover, the mean of $X_{(k:n)}$ 
	can be expressed in terms of the CIGF of $X$ as 
	$$
	\mathbb{E}\left[X_{(k:n)}\right] = \sum_{j=0}^{k-1}\binom{n}{j}\int_{l}^{r}\left[F(x)\right]^j\left[\overline{F}(x)\right]^{n-j}\,{\rm{d}}x
	=\sum_{j=0}^{k-1}\binom{n}{j}\,G_X(j,n-j).
	$$
	Clearly, for $k=1$  we have $\mathbb{E}\left[X_{(1:n)}\right]=G_X(0,n)=K_X(n)$. 
\end{remark}
%

\subsection{Stress-strength models for multi-component systems}\label{section-Stress-strength}
A further connection of the CIGF with systems reliability arises in the analysis of stress-strength models for multi-component systems. 
Let  us consider a system with $n$ components having i.i.d.\  strengths $X_1, X_2, \ldots, X_n$ having common CDF $F(x)$. 
Assume that each component is stressed  according to an independent 
random stress $T$ having   CDF $F_T(x)$. Moreover, suppose that the system survives if and only if the 
components strengths are greater than the stress by at least $k$ out of $n$ ($1\leq k\leq n$). 
Then, the reliability of the considered  multi-component stress-strength system is given by 
(cf.\ Bhattacharyya and Johnson \cite{Bhatta}) 
\begin{eqnarray}
	R_{k,n} & = & \mathbb{P}[\hbox{at least $k$ of the $(X_1, X_2, \ldots, X_n)$ exceed $T$}]
	\nonumber 
	\\
	& = & \sum_{j=k}^n {n\choose j}\int_{-\infty}^{+\infty} [1-F(t)]^j [F(t)]^{n-j}\, {\rm d}F_T(t),
	\label{eq:Rkn}
\end{eqnarray}
with $R_{0,n} =1$. 
See also, for instance, the recent contribution by Kohansal and Shoaee \cite{Kohansal} on 
the statistical inference of multicomponent stress-strength reliability under suitable censored samples. 
\par
For instance, if $T$ is distributed as $X_1$ then it is not hard to see that 
\begin{equation}
	R_{k,n} = 1-\frac {k}{n+1},\qquad 0\leq k\leq n. 
	\label{eq:linRkn}
\end{equation}
The following result is a straightforward consequence of Eqs.\  (\ref{G_X})  and (\ref{eq:Rkn}). 
\begin{proposition}\label{prop:Rkn}
	Let $X_1, X_2, \ldots, X_n,T$ have common finite support $(l,r)$, with $X_1, X_2, \ldots, X_n$ i.i.d. 
	If $T$ is uniformly distributed over $(l,r)$, then 
	$$
	R_{k,n} =\frac{1}{r-l} \sum_{j=k}^n {n\choose j}G_X(n-j,j), \qquad 0\leq k\leq n. 
	$$
\end{proposition}
\par
An iterative formula allows us to evaluate  the reliability of the    multi-component stress-strength system 
as follows, under the assumptions of Proposition \ref{prop:Rkn}: 
$$
R_{k+1,n} =R_{k,n} -\frac{1}{r-l}  {n\choose k}G_X(n-k,k), \qquad 0\leq k\leq n-1. 
$$
\par
As example, if $X$ has Power$(\theta)$ distribution then from Proposition \ref{prop:Rkn} and Table \ref{tableCIGF} 
after few calculations one has 
\begin{equation}
	R_{k,n} =\frac{ \Gamma(n + 1)\,\Gamma(n - k + 1 + \frac{1}{\theta})}
	{\Gamma(n + 1 + \frac{1}{\theta})\,\Gamma(n - k + 1)},
	\qquad 0\leq k\leq n, \quad \theta>0.
	\label{eq:powerRkn}
\end{equation}
Note that the expression in (\ref{eq:powerRkn}) can be also represented as a ratio of Pochhammer symbols, 
or as an infinite product (cf.\ Eq.\ 8.325.1 of Gradshteyn and  Ryzhik \cite{Gradshteyn}). Clearly, 
if $\theta=1$ then  Eq.\ (\ref{eq:powerRkn}) reduces to Eq.\ (\ref{eq:linRkn}). 
Figure \ref{Fig:Rkn} shows some plots of  $R_{k,n}$ as given in (\ref{eq:powerRkn}). 
\begin{figure}[t] 
	\begin{center}
		\includegraphics[scale=0.4]{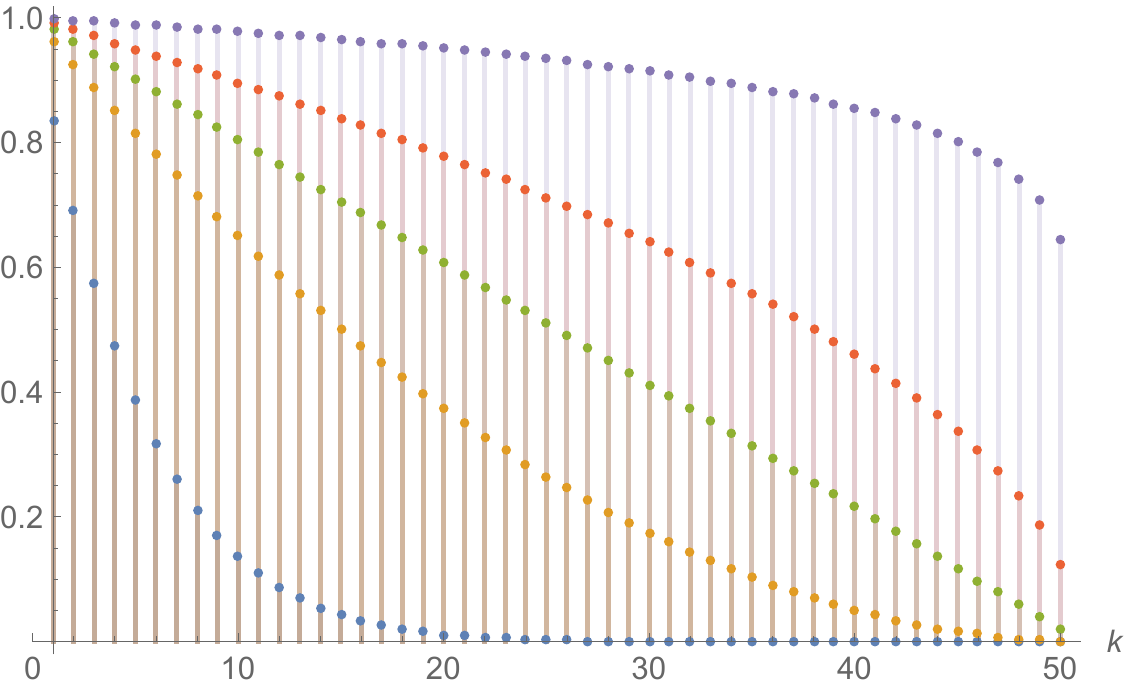}
		\,
		\includegraphics[scale=0.4]{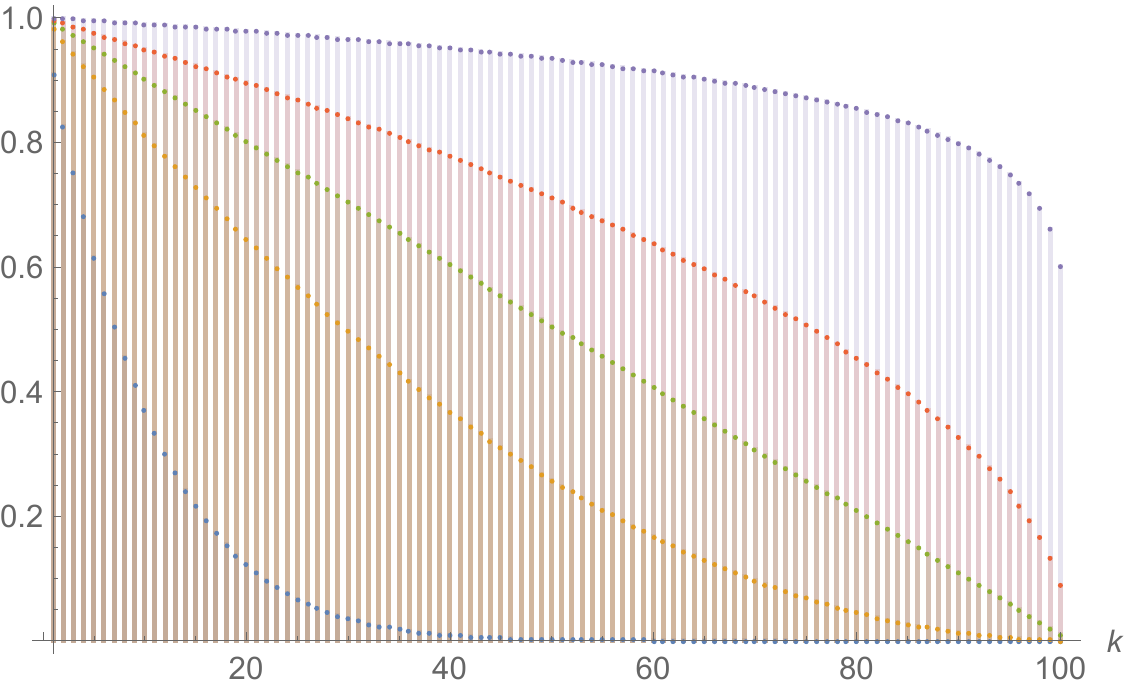}
	\end{center}
	\caption{Plots of the reliability of the  multi-component stress-strength system with underlying Power$(\theta)$ distribution 
		as in Eq.\ (\ref{eq:powerRkn}) for $0\leq k\leq n$, with $n=50$ (left) and $n=100$ (right), with $\theta=0.1, 0.5, 1, 2, 10$ (from bottom to top).}
	\label{Fig:Rkn}
\end{figure}

\section{Generalized Gini functions}\label{sect:Generalizations}
This section is devoted to the analysis of a generalized version of the CIGF. 
Specifically, we aim to extend the Definition \ref{def G_X} to the case in which the powers included in the 
right-hand-side of Eq.\  (\ref{G_X}) are replaced by suitable distortion functions. 
In this case, recalling Remark \ref{rem:Gini}, we come to an extension of the Eq.\ (\ref{eq:GiniX}). 
\par
Let $X$ be a random variabile with CDF $F$ and SF $\overline{F}$, and let 
$q_i:[0,1]\rightarrow[0,1]$ be two {\em distortion functions}, for $i=1,2$, 
i.e.\ increasing functions such that $q_i(0)=0$ and $q_i(1)=1$ (cf.\ Section 2.9.2 of Belzunce et al.\ \cite{Belzunce:2016} 
and Section 2.4 of Navarro \cite{Navarro:2022}). In some applications it is required that $q_i$ is continuous or left-continuous, however 
these assumptions are not necessarily required in general. 
The distorted distribution function and the distorted survival function of $F$ through $q_i$ are given respectively by
\begin{equation}
F_{q_1}(x)=q_1\left(F(x)\right), \qquad \overline{F}_{q_2}(x)=q_2\left( \overline F(x)\right), 
\qquad\forall x\in\mathbb{R}.
 \label{eq:distfunct}
\end{equation}
From Eq.\ (\ref{eq:distfunct}), in general one has $F_{q_1}(x)+\overline{F}_{q_2}(x)\neq 1$, for all $x\in(l,r)$, 
unless $q_1(u)=1-q_2(1-u)$. 
The functions given in  (\ref{eq:distfunct}) have been introduced in the context of the
theory of choice under risk (see Wang \cite{Wang:1996}), and are largely used in various applied fields  (for instance, see Sordo and  Su\'arez-Llorens  \cite{Sordo:2011} 
for applications to variability measures). 
\par 
Let us now consider the preannounced generalization of the CIGF based on  (\ref{eq:distfunct}).
\begin{definition}\label{GCIGF}
	Let $X$ be a  random variable with CDF $F(x)$ and SF $\overline{F}(x)$, $x\in\mathbb{R}$, and let  
\begin{equation*}
	l=\inf\{x\in\mathbb{R} : F(x)>0\}, \qquad r=\sup\{x\in\mathbb{R} : \overline{F}(x)>0\}.  
\end{equation*}
The  {\bf q}-distorted Gini function (or, shortly, {\bf q}-Gini function)  of $X$ is defined as
\begin{equation}\label{hatGX}	
		\hat{G}_X({\bf q})=\int_{l}^{r} F_{q_1}(x) \overline{F}_{q_2}(x)\,{\rm{d}}x,	
\end{equation}
where ${\bf q}=(q_1,q_2)$, and   $q_i:[0,1]\rightarrow[0,1]$, $i=1,2$,  are  distortion functions such that 
$\hat{G}_X({\bf q})$ is finite. 
\end{definition}
\par
Clearly, if the distortion functions are taken as $q_1(x)=x^\alpha$ and $q_2(x)=x^\beta$ with $(\alpha,\beta)\in D_X$, then 
Eq.\ (\ref{hatGX}) corresponds to the definition of the CIGF given in Eq.\ (\ref{G_X}). Specifically, if $\alpha=\beta=1$    
then we recover the Gini mean semi-difference (\ref{eq:GiniX}). 
\par
It is worth mentioning that the {\bf q}-Gini function may be viewed as an extension of the distorted measures treated in 
Giovagnoli and Wynn \cite{Giovagnoli} and in Greselin and Zitikis \cite{Greselin}. In these papers, distortions of the only CDF or SF (which can be viewed as generalizations of Eqs.\ (\ref{H_X}) and (\ref{CIG})), are 
considered for the analysis of stochastic dominance, Lorenz ordering and risk measures. 
\par
Hereafter we shall prove various results regarding the {\bf q}-Gini function. 
They include the effect of an affine transformation of $X$ and the pointwise ordering of the {\bf q}-Gini functions. 
To this aim, we recall that, if $X$ and $Y$ are nonnegative  random variables with CDFs $F$ and $G$, respectively, 
then  $X$ is said to be smaller than $Y$ in the dispersive order, denoted as $X\le_d Y$, if and only if 
(see Section 3.B of Shaked and Shanthikumar \cite{shaked:2007}) 
\begin{equation*}
	F^{-1}(v)-F^{-1}(u) \ge G^{-1}(v)- G^{-1}(u)\qquad
	\hbox{whenever }0<u \leq v<1. 
\end{equation*}
Among the variability stochastic orders, the dispersive order is one of the most popular, since it involves 
quantities that are easily tractable, requiring that the difference
between any two quantiles of $X$ is smaller than the corresponding quantity of $Y$. Moreover, 
if $X$ and $Y$ are absolutely continuous with PDFs $f$ and $g$, respectively, then  
\begin{equation}\label{eq:disp}
	X\le_d Y     \quad \hbox{if and only if} \quad f(F^{-1}(u))\ge g(G^{-1}(u))\quad\forall \, u\in(0,1).
\end{equation}
We can now prove that, under suitable assumptions, the {\bf q}-Gini function is a variability measure 
in the sense of Bickel and Lehmann \cite{Bickel:Lehmann}.
%

\begin{theorem}\label{th:confr}
Let $X$ and $Y$ be random variables having the same support, 
and let  ${\bf q}=(q_1,q_2)$, where  $q_i:[0,1]\rightarrow[0,1]$, $i=1,2$,  be  distortion functions such that 
$\hat{G}_X({\bf q})$ and $\hat{G}_Y({\bf q})$  are finite. 
Then, the following properties hold:
	\begin{enumerate}
		\item $\hat{G}_{X+\delta}({\bf q})=\hat{G}_X({\bf q})$ for all $\delta\in \mathbb R$;
		\item $\hat{G}_{\gamma X}({\bf q})=\gamma\hat{G}_X({\bf q})$ for all $\gamma\in \mathbb{R}^+$;
		\item $\hat{G}_X({\bf q})=0$ for any degenerate random variable $X$;
		\item $\hat{G}_X({\bf q})\ge0$ for any random variable $X$;
		\item $X\le_d Y$ implies $\hat{G}_X({\bf q})\le \hat{G}_Y({\bf q})$.
	\end{enumerate}
	\end{theorem}
\begin{proof}
The properties 1 and 2 follow recalling  Eq.\ (\ref{hatGX}) 
and  the relation between the CDFs of $X$ and $Y=\gamma X + \delta$. 
The properties 3 and 4 are  guaranteed by the Definition \ref{GCIGF}. 
In analogy to  Eq.\ (\ref{CIGF in Remark 1}), we can write 
	\begin{equation*}
		\hat{G}_Y({\bf q})-\hat{G}_X({\bf q})
		=\int_{0}^{1}q_1(u)q_2(1-u)\left[
		\frac{1}{g(G^{-1}(u))}-\frac{1}{f(F^{-1}(u))}\right]\,\text{d}u.
	\end{equation*}
Hence, due to relation (\ref{eq:disp}) the property 5 immediately follows.
\end{proof}
\par
In addition, we remark that if $\hat{G}_X({\bf q})=0$ and if 
$q_i(0^+)>0$ and $q_i(1^-)<1$ for $i=1,2$, then $X$ is necessarily a degenerate random variable.
\par
It is worth mentioning that the results given in this section can be further extended. Indeed, in the same conditions given in the 
Definition \ref{GCIGF}, by taking as reference the approach in \cite{Giovagnoli} we   introduce the 
{\em weighted {\bf q}-distorted Gini function} (or, shortly, {\em weighted {\bf q}-Gini function}) as follows:
 \begin{equation}\label{extendhatGX}	
		\hat{G}_X({\bf q},F_T)=\int_{\Delta} F_{q_1}(x) \overline{F}_{q_2}(x)\,{\rm{d}}F_T(x),	
\end{equation}
where $F_T(x)$ is the CDF of a random variable $T$, 
and where the intersection of the supports of $X$ and $T$ is a non-empty set  denoted by $\Delta$. 
\par
It is not hard to see that the function given in (\ref{extendhatGX}) satisfies the  properties 1--4 given 
in Theorem \ref{th:confr} for $\hat{G}_X({\bf q})$. Concerning the property 5, hereafter we see that 
additional assumptions are needed. Here, $l_X$, $r_X$ and $l_Y$, $r_Y$ are defined as in (\ref{eq:deflr}) 
for $X$ and $Y$, respectively. 
\begin{theorem}\label{th:wconfr}
Let $X$ and $Y$ be random variables having the same support, 
let  ${\bf q}=(q_1,q_2)$, where  $q_i:[0,1]\rightarrow[0,1]$, $i=1,2$,  be  distortion functions such that 
$\hat{G}_X({\bf q},F_T)$ and $\hat{G}_Y({\bf q},F_T)$ are finite, and let $T$ be absolutely continuous with PDF $f_T$. If 
\\
(i) $f_T(t)$ is increasing in $t$ and $-\infty<l_X=l_Y$, or if 
\\
(ii) $f_T(t)$ is decreasing in $t$ and $ r_X=r_Y<\infty$, then 
\begin{equation}\label{eq:thXdY}
 X\le_d Y \quad\text{implies}\quad  \hat{G}_X({\bf q},F_T)\le \hat{G}_Y({\bf q},F_T).
\end{equation}
\end{theorem}
\begin{proof}
Due to (\ref{extendhatGX}), by setting $u = F (x)$  one has 
 \begin{equation*}
		\hat{G}_Y({\bf q},F_T)-\hat{G}_X({\bf q},F_T)
		=\int_{0}^{1}q_1(u)q_2(1-u)\left[
		\frac{f_T(G^{-1}(u))}{g(G^{-1}(u))}-\frac{f_T(F^{-1}(u))}{f(F^{-1}(u))}\right]\,\text{d}u.
 \end{equation*}
Then, under assumption (i), from Theorem 3.B.13 of  \cite{shaked:2007} we have that 
assumption $X\le_d Y$ implies $X\le_{st} Y$, 
i.e.\ $F(x)\geq G(x)$ for all $x\in \mathbb R$, so that $G^{-1}(u)\geq F^{-1}(u)$ for all $u\in (0,1)$. The relation (\ref{eq:thXdY}) thus 
follows from (\ref{eq:disp}). The same result can be proved similarly under assumption (ii). 
\end{proof}
An immediate application of  Theorem \ref{th:wconfr} can be given to the 
reliability of  multi-component stress-strength systems, as seen in Section \ref{section-Stress-strength}.
Consider two $n$-component systems, the first  having i.i.d.\  strengths $X_1, X_2, \ldots, X_n$ distributed as $X$, 
and the second having i.i.d.\  strengths $Y_1, Y_2, \ldots, Y_n$ distributed as $Y$. 
Assume that each component of both systems is stressed  according to an independent 
random stress $T$ having   CDF $F_T(x)$. We denote by $R_{k,n}^X$ and $R_{k,n}^Y$   
the reliability of the corresponding multi-component stress-strength systems defined as in 	(\ref{eq:Rkn}). 
We are now able to provide a comparison result based on the weighted {\bf q}-Gini function. 
\begin{theorem}
Let the strengths $X$ and $Y$ have the same support, and let the random stress $T$ be absolutely continuous 
with PDF $f_T$. If 
\\
(i) $f_T(t)$ is increasing in $t$ and $-\infty<l_X=l_Y$, or if 
\\
(ii) $f_T(t)$ is decreasing in $t$ and $ r_X=r_Y<\infty$, then 
\begin{equation*} 
 X\le_d Y \quad\text{implies}\quad  R_{k,n}^X\leq R_{k,n}^Y  \quad  \text{for all}\quad   0\leq k\leq n,
\end{equation*}
provided that $R_{k,n}^X$ and $R_{k,n}^Y$ are finite. 
\end{theorem}
\begin{proof}
The thesis follows recalling Eq.\ (\ref{eq:Rkn}) and making use of Theorem \ref{th:wconfr} when 
the relevant distortions are given by $q_1(u)=u^{n-j}$ and $q_2(u)=u^{j}$, with $k\leq j\leq n$. 
\end{proof}
\par
In the last result of this section, thanks to the probabilistic analogue of the mean value theorem, 
we provide a suitable expression of the weighted {\bf q}-Gini function in the special case  when 
the related distortion functions are equal to the identity. 
\begin{proposition}
Let $X$ be a nondegenerate random variable such that $\mathbb{E}(\min\{X,X'\})$ and $\mathbb{E}(\max\{X,X'\})$ 
are finite, where $X'$ is an independent copy of $X$. 
For the weighted {\bf q}-Gini function introduced in Eq.\ (\ref{extendhatGX}), 
if $T$ is absolutely continuous with the same support of $X$, and if 
$q_i(u)=u$, $0\leq u\leq 1$, for $i=1,2$, then  
   \begin{equation}\label{extendhatGX_prop}	
	\hat{G}_X({\bf q},F_T)=\frac{1}{2}\,\mathbb{E}\left[F_T(\max\{X,X'\})-F_T(\min\{X,X'\})\right].
   \end{equation}
\end{proposition} 
\begin{proof}
	The proof follows  making use of Eq.\ (\ref{extendhatGX}) 
	and   Theorem 4.1 in Di Crescenzo \cite{dicre:1999} extended to the case of a general support of $X$, 
	by taking $Z=\Psi(\min\{X,X'\},\max\{X,X'\})$ and   $g(\cdot)=F_T(\cdot)$. 
\end{proof}
We immediately note that Eq.\ (\ref{extendhatGX_prop}) generalizes Eq.\ (\ref{eq:GiniX}) in this proposed context. 

\section{Two-dimensional cumulative information generating function}\label{section bidimensional CIGF}
Let us now extend the analysis of the CIGF to the case of a two-dimensional random vector.
In analogy with Definition \ref{def G_X}, by avoiding trivial degenerate cases, we introduce the following 
\begin{definition}\label{def G_(X,Y)}
	Let $(X,Y)$ be random vector with nondegenerate components, having joint CDF and SF given respectively by 
	$$
	F(x,y)=P(X\le x,Y\le y), \qquad \overline{F}(x,y)=P(X>x,Y>y), \qquad (x,y)\in\mathbb{R}^2. 
	$$
	We consider the following domain 
	$$
	\mathcal{S}_{(X,Y)}=\{(x,y)\in\mathbb{R}^2 : F(x,y)\overline{F}(x,y)>0\}.
	$$
	The  CIGF of $(X,Y)$  is defined as:
	$$
	\begin{aligned}
		G_{(X,Y)}:&\:D_{(X,Y)}\subseteq \mathbb{R}^{2} \longrightarrow (0,+\infty)\\
		&\qquad(\alpha,\beta)\;\quad\;\longmapsto \;G_{(X,Y)}(\alpha,\beta)= \iint_{\mathcal{S}_{(X,Y)}}[F(x,y)]^{\alpha}\:[\overline{F}(x,y)]^{\beta}\,{\rm{d}}x\,{\rm{d}}y
	\end{aligned}
	$$
	where
	$$
	D_{(X,Y)}=\{(\alpha,\beta)\in\mathbb{R}^2 : G_{(X,Y)}(\alpha,\beta)<+\infty\}.
	$$
\end{definition}
\par
We first discuss few examples. The first example is stimulated by  the fact that if 
$X\sim \text{Bernoulli}\left(\frac{1}{2}\right)$, then  $G_X(\alpha,\beta)=\left(\frac{1}{2}\right)^{\alpha+\beta}$ 
(cf.\  Table \ref{tableCIGF}), thus satisfying the symmetry conditions expressed in Remark \ref{rem:symm}.   
\begin{example}\label{discrete example biCIGF}
	Let $(X,Y)$ be a discrete random vector, with probability function 
	\begin{equation*}
		\mathbb{P}(X=x,Y=y)=
		\begin{cases}
			\frac14+\theta\qquad(x,y)\in \{(0,0),(1,1)\}\\
			\frac14-\theta\qquad(x,y)\in \{(0,1),(1,0)\}\\
			0\qquad\qquad\text{otherwise}
		\end{cases}
	\end{equation*}
	for $\theta\in\left(-\frac14,\frac14\right)$. Therefore the CDF and the SF are identical, given by 
	$F(x,y)=\overline{F}(x,y)=\frac14+\theta$,
	for $(x,y)\in \mathcal{S}_{(X,Y)}=[0,1]^2$. 
	Hence, from Definition \ref{def G_(X,Y)} the CIGF of $(X,Y)$ is 
	\begin{equation*}
		G_{(X,Y)}(\alpha,\beta)
		=\left(\frac14+\theta\right)^{\alpha+\beta},
		\qquad (\alpha,\beta)\in D_{(X,Y)}=\mathbb{R}^2.
	\end{equation*}
\end{example}
\begin{example}
	Let $(X,Y)$ be an absolutely continuous random vector, uniformly distributed in the triangular domain 
	$\mathcal{T}=\left\lbrace (x,y)\in\mathbb{R}^2: 0\le x\le 1,\,  0\le y\le 1-x\right\rbrace$. 
	The PDF,   CDF and SF are given respectively by 
	\begin{equation*}
		f(x,y)=2, \qquad 
		F(x,y)=2xy, \qquad 
		\overline{F}(x,y)=(1-x-y)^2, 
		\qquad 
		\hbox{for $(x,y)\in \mathcal{T}$.}
	\end{equation*}
	In this case $\mathcal{S}_{(X,Y)}=\mathcal{T}$, so that the CIGF of $(X,Y)$ is 
	\begin{equation*}
		\begin{aligned}
			G_{(X,Y)}(\alpha,\beta)&= \int_{0}^{1}\int_{0}^{1-x}\left(2xy\right)^{\alpha}\:\left(1-x-y\right)^{2\beta}\:{\rm{d}}y\,{\rm{d}}x\\
			&=2^\alpha B(\alpha+1,2\beta+1)\int_{0}^{1}x^\alpha(1-x)^{\alpha+2\beta+1}\,{\rm{d}}x\\
			&=2^\alpha B(\alpha+1,2\beta+1)B(\alpha+1, \alpha+2\beta+2),
		\end{aligned}
	\end{equation*}
	with 
	\begin{equation*}
		D_{(X,Y)}=\left\lbrace(\alpha,\beta)\in\mathbb{R}^2 : \alpha>-1, \beta>-\frac{1}{2}\right\rbrace.
	\end{equation*}
\end{example}
\begin{example}
	Let $(X,Y)$ be an absolutely continuous random vector,   distributed on the   domain 
	$\mathcal{Q}=[0,1]^2$, with PDF, CDF and SF given respectively by,  for $(x,y)\in\mathcal{Q}$, 
	$$
	f(x,y)=x+y,\quad 
	F(x,y)=\frac12 xy(x+y),\quad 
	\overline{F}(x,y)=\frac12(x-1)(y-1)(x+y+2),
	$$ 	
	so that in this case $\mathcal{S}_{(X,Y)}=\mathcal{Q}$. 
	Since $\left\vert\frac{x+y}{2}\right\vert<1$ for $(x,y)\in\mathcal{Q}$, 
	%
	%
	recalling the expression of the Gauss hypergeometric function (cf.\ 15.3.1 of Abramowitz and Stegun \cite{Abramowitz:Stegun}) 
	$$
	{}_{2}{F_1(a,b,c,z)}=\frac{\Gamma(c)}{\Gamma(b)\Gamma(c-b)}\int_{0}^{1}t^{b-1}(1-t)^{c-b-1}(1-tz)^{-a}\,{\rm{d}}t,
	\qquad Re(c)>Re(b)>0,
	$$
	and making use of 15.3.7 in \cite{Abramowitz:Stegun} and of 2.21.1.4 in 
	Prudnikov et al.\ \cite{Pruonikov}, the CIGF of $(X,Y)$ is
	\begin{equation*}
		\begin{aligned}
			G_{(X,Y)}(\alpha,\beta)&=\frac{1}{2^\alpha}\sum_{k=0}^{+\infty}{\beta\choose k}\frac{1}{2^k}\frac{\Gamma(\beta+1)\Gamma(2\alpha+k+1)}{\Gamma(2+2\alpha+\beta+k)}B(\alpha+1,\beta+1)
			\\
			&\quad \times  {}_{3}{F_2(-\alpha-k,-1-2\alpha-\beta-k,\alpha+1,-2\alpha-k,-\alpha+\beta+2,-1)}  \\
			&+\frac{1}{2^\alpha}\sum_{k=0}^{+\infty}{\beta\choose k}\frac{1}{2^k}\frac{\Gamma(\alpha+1)\Gamma(-2\alpha-k-1)}{\Gamma(-\alpha-k)}B(3\alpha+k+2,\beta+1) \\
			&\quad \times {}_{3}{F_2(\alpha+1,-\beta,3\alpha+k+2,2+2\alpha+k,3\alpha+\beta+3,-1)}
		\end{aligned}
	\end{equation*}
	with
	$D_{(X,Y)}=\{(\alpha,\beta)\in\mathbb{R}^2:\alpha,\beta\in\mathbb{R}\setminus \mathbb{Z}_0^-\}$, 
	where the function ${}_{3}F_2$ can be found in \cite{Pruonikov}, for instance. 
\end{example}
\par
The following result is an immediate consequence of the involved notions. 
\begin{proposition}
	Let $(X,Y)$ be a random vector having finite CIGF. 
	If $X$ and $Y$ are independent then
	\begin{equation*}
		G_{(X,Y)}(\alpha,\beta)=G_{X}(\alpha,\beta)\:G_{Y}(\alpha,\beta) \qquad \forall\,(\alpha,\beta)\in D_{(X,Y)}.
	\end{equation*}
\end{proposition}
\par
Consider a nonnegative random vector $(X,Y)$ with support $(0, r_1)\times(0, r_2)$, for $r_1,r_2\in(0,+\infty]$. 
In many practical situations, it is worthwhile to adopt  the following information measures for multi-device systems. 
The \emph{joint cumulative residual entropy} of $(X,Y)$ is defined as (cf.\  \cite{rao:etal}) 
\begin{equation}\label{CRE_(X,Y)}
	\mathcal{CRE}(X,Y)
	=-\int_{0}^{r_1}\text{d}x\int_{0}^{r_2}\overline{F}(x,y)\log\overline{F}(x,y) \,\text{d}y.
\end{equation}
A dynamic version of this measure has been studied by Rajesh et al.\ \cite{Rajesh:2014}. 
Similarly, the \emph{joint cumulative entropy} of $(X,Y)$ is defined as (cf.\ \cite{dicre:longo}) 
\begin{equation}\label{CE_(X,Y)}
	\mathcal{CE}(X,Y)
	=-\int_{0}^{r_1}\text{d}x\int_{0}^{r_2}F(x,y)\log F(x,y) \,\text{d}y.
\end{equation}
\begin{remark}
	We recall that, if $(X,Y)$  has support  $(0, r_1)\times(0, r_2)$, for $r_1,r_2\in(0,+\infty)$, 
	and if $X$ and $Y$ are independent, then (cf.\ Proposition 2.2 of \cite{dicre:longo}) 
	\begin{equation}\label{CE_(X,Y) when X,Y are independent}
		\mathcal{CE}(X,Y)=\left[r_2-\mathbb{E}(Y)\right]\mathcal{CE}(X)+\left[r_1-\mathbb{E}(X)\right]\mathcal{CE}(Y).
	\end{equation}
	Under the same assumptions, similarly, from Eq.\ (\ref{CRE_(X,Y)}) 
	we can observe that 
	\begin{equation}\label{CRE_(X,Y) when X,Y are independent}
		\mathcal{CRE}(X,Y)=\mathbb{E}\left(Y\right)\mathcal{CRE}(X)+\mathbb{E}\left(X\right)\mathcal{CRE}(Y).
	\end{equation}
\end{remark}
\begin{remark}
	For the discrete random vector  $(X,Y)$ considered in 
	Example \ref{discrete example biCIGF}, we have 
	$$
	\mathcal{CRE}(X)=\mathcal{CRE}(Y)=\mathcal{CE}(X)=\mathcal{CE}(Y)=-\frac12\log\left(\frac12\right)
	$$ 
	and 
	$$
	\mathcal{CRE}(X,Y)=\mathcal{CE}(X,Y)=-\left(\frac14+\theta\right)\log\left(\frac14+\theta\right), 
	\qquad \theta\in\left(-\frac14,\frac14\right).
	$$ 
	Hence, in this case  Eqs.\ (\ref{CE_(X,Y) when X,Y are independent}) and (\ref{CRE_(X,Y) when X,Y are independent}) are satisfied if and only if $X$ and $Y$ are independent, i.e.\  $\theta=0$. 
\end{remark}
\par
In analogy with the one-dimensional measures considered in Table \ref{Entropies}, 
we define the generalized and the 
fractional versions of the measures given in Eqs.\ (\ref{CRE_(X,Y)}) and (\ref{CE_(X,Y)}). 

\begin{definition}\label{def CRE_n(X,Y) and CE_n(X,Y)}
	Let $(X,Y)$ be a nonnegative random vector with support $(0,r_1)\times(0,r_2)$, 
	where $r_1,r_2\in(0,+\infty]$. 
	The generalized cumulative residual entropy of order $n$ of $(X,Y)$ is defined as
	\begin{equation*}
		\mathcal{CRE}_n(X,Y)
		=\frac{1}{n!}\int_{0}^{r_1}{\rm{d}}x\int_{0}^{r_2}\overline{F}(x,y)\left[-\log\overline{F}(x,y)\right]^n {\rm{d}}y,
		\qquad n\in \mathbb{N}_0,
	\end{equation*} 
	while the generalized cumulative entropy of order $n$ of $(X,Y)$ is defined as
	\begin{equation*}
		\mathcal{CE}_n(X,Y)=\frac{1}{n!}\int_{0}^{r_1}{\rm{d}}x\int_{0}^{r_2}F(x,y)\left[-\log F(x,y)\right]^n {\rm{d}}y,
		\qquad n\in \mathbb{N}.
	\end{equation*}
\end{definition}
\begin{definition}\label{def fractional CRE(X,Y) and CE(X,Y)}
	Let $(X,Y)$ be a nonnegative random vector with support $(0,r_1)\times(0,r_2)$, with 
	$r_1,r_2\in(0,+\infty]$. The fractional cumulative residual entropy of $(X,Y)$ is defined as 
	\begin{equation*}
		\mathcal{CRE}_\nu(X,Y)
		=\frac{1}{\Gamma(\nu+1)}\int_{0}^{r_1}{\rm{d}}x\int_{0}^{r_2}\overline{F}(x,y)\left[-\log\overline{F}(x,y)\right]^\nu {\rm{d}}y,
		\qquad\nu\ge0
	\end{equation*}
	whereas the fractional cumulative entropy of $(X,Y)$ is defined as
	\begin{equation*}
		\mathcal{CE}_\nu(X,Y)
		=\frac{1}{\Gamma(\nu+1)}\int_{0}^{r_1}{\rm{d}}x\int_{0}^{r_2}F(x,y)\left[-\log F(x,y)\right]^\nu {\rm{d}}y,
		\qquad \nu>0.
	\end{equation*}
\end{definition}
\par
According to Propositions \ref{prop CE_X and CRE_X from G_X}, 
\ref{prop CE_n and CRE_n from G_X} and \ref{prop fractional CRE CE from CIGF }, 
now we propose the following generalizations, whose proof is analogous. 
\begin{proposition}
	Let $(X,Y)$ be a random vector with finite CIGF $G_{(X,Y)}(\alpha,\beta)$, for 
	$(\alpha,\beta)\in D_{(X,Y)}$, and let 
	${\mathcal{S}_{(X,Y)}}=(0,r_1)\times(0,r_2)$ where $r_1,r_2\in(0,+\infty]$. 
	If $(0,1)\in D_{(X,Y)}$, then
	\begin{equation*}
		\mathcal{CRE}(X,Y)=-\frac{\partial}{\partial\beta}G_{(X,Y)}(\alpha,\beta)\Big\vert_{\alpha=0,\beta=1},
	\end{equation*}	
	\begin{equation*}
		\mathcal{CRE}_n(X,Y)
		= \frac{(-1)^n}{n!} \frac{\partial^n}{\partial\beta^n}G_{(X,Y)}(\alpha,\beta)\Big\vert_{\alpha=0,\beta=1}, 
		\qquad n\in \mathbb{N}_0,
	\end{equation*}
	\begin{equation*}
		\mathcal{CRE}_\nu(X,Y)=\frac{1}{\Gamma(\nu+1)}\left(\prescript{C}{}{D_{-,\beta}^\nu G_{(X,Y)}}\right)(\alpha,\beta)\Big\vert_{\alpha=0,\beta=1},
		\qquad \nu> 0,
	\end{equation*}
	If $(1,0)\in D_{(X,Y)}$, then
	\begin{equation*}
		\mathcal{CE}(X,Y)=-\frac{\partial}{\partial\alpha}G_{(X,Y)}(\alpha,\beta)\Big\vert_{\alpha=1,\beta=0}, 
	\end{equation*}
	\begin{equation*}
		\mathcal{CE}_n(X,Y)= \frac{(-1)^n}{n!} \frac{\partial^n}{\partial\alpha^n}G_{(X,Y)}(\alpha,\beta)\Big\vert_{\alpha=1,\beta=0}, 
		\qquad n\in \mathbb{N},
	\end{equation*}
	\begin{equation*}
		\mathcal{CE}_\nu(X,Y)=\frac{1}{\Gamma(\nu+1)}\left(\prescript{C}{}{D_{-,\alpha}^\nu G_{(X,Y)}}\right)(\alpha,\beta)\Big\vert_{\alpha=1,\beta=0},
		\qquad \nu>0.
	\end{equation*}
\end{proposition}
\par
In analogy with Eq.\ (\ref{Odds Function}), 
if $(X,Y)$ is a random vector with CDF $F(x,y)$ and SF $\overline{F}(x,y)$, 
for all $(x,y)\in\mathcal{S}_{(X,Y)}$ 
we can  introduce the two-dimensional   \emph{odds function}   as
\begin{equation}\label{eq:odd2d}
	\theta(x,y)=\frac{\overline{F}(x,y)}{F(x,y)}.
\end{equation}
Hence, the two-dimensional extension  of  Eq.\ (\ref{G_X in terms of only Odds Function with beta})  
can be given   in terms of the function  in (\ref{eq:odd2d}), so  that
\begin{equation*}
	G_{(X,Y)}(-\beta,\beta)=\iint_{\mathcal{S}_{(X,Y)}}[\theta(x,y)]^{\beta}\:\text{d}x\:\text{d}y
\end{equation*}
for all  $\beta \in \mathbb{R}$ such that the right-hand-side is finite. 
\section{Concluding remarks}\label{sect:final}
In this paper, we defined the cumulative information generating function and its suitable distortions-based extensions. 
The CIGF is noteworthy in the context of information measures, since it allows to determine the cumulative residual entropy and the cumulative entropy of a given probability distribution, even in their generalized and fractional forms. 
\par
Several results, properties and bounds  have been studied, also with reference to symmetry properties 
and relations with the equilibrium density. 
Moreover, we illustrated that the considered   functions  are variability measures 
which extend the Gini mean semi-difference. 
\par
In the realm of reliability theory, the CIGF and its extensions have been found useful to study 
properties of $k$-out-of-$n$ systems, with special regards to 
the  reliability of   multi-component stress-strength systems. 
\par
Future developments can be oriented to connections with other 
notions, such as the Information Divergence  
(see Toulias and Kitsos \cite{toulias:kitsos}) 
and the R\'enyi Entropy (see, for instance, Buryak and Mishura \cite{Buryak:Mishura}), 
also with reference to possible applications and distortions-based extensions 
of the two-dimensional CIGF introduced in Section \ref{section bidimensional CIGF}. 
%
\section*{Credit Author Statement} 
All authors have contributed equally to this work.

\section*{Declaration of Competing Interest} 
The authors declare that they have no known competing financial interests or personal relationships that could have appeared to influence the work reported in this paper.

\section*{Acknowledgements} 
The authors are members of the group GNCS of INdAM (Istituto Nazionale di Alta Matematica).
This work is partially supported by MIUR--PRIN 2017, Project ``Stochastic Models for Complex Systems'' (no. 2017JFFHSH).


\end{document}